\documentclass[sigconf]{acmart}

\setcopyright{acmlicensed}
\acmConference[KDD '26]{Proceedings of the 32nd ACM SIGKDD Conference on Knowledge Discovery and Data Mining}{August 2026}{Jeju, Korea}

\AtBeginDocument{%
  }

\usepackage{pifont}  
\usepackage{booktabs} 
\usepackage{enumitem}
\usepackage{multirow} 
\usepackage{svg}
\usepackage{titletoc}
\usepackage{subcaption}
\usepackage{tikz}
\usepackage{hyperref}

\begin{document}

\title{\mname: Structure Elucidation from Experimental NMR Spectra via Set Transformers}

\newcommand{\mname}{NMRTrans}
\newcommand{\dataset}{NMRSpec}

\newcommand{\cmark}{\ding{51}} 
\newcommand{\xmark}{\ding{55}} 

\renewcommand{\authors}{Liujia Yang$^\ast$, Zhuo Yang$^\ast$, Jiaqing Xie$^\ast$, Yubin Wang$^\ast$, Ben Gao, Xingjian Wei, Jiaxing Sun, Tianfan Fu, Jiang Wu$^\dagger$, Conghui He$^\dagger$, Yuqiang Li$^\dagger$, Qinying Gu$^\dagger$}

\author{Liujia Yang$^\ast$\\Zhuo Yang$^\ast$} 
\email{{yangliujia, yangzhuo}@pjlab.org.cn}
\affiliation{\institution{Shanghai Artificial Intelligence Laboratory}\city{Shanghai}\country{China}}

\author{Jiaqing Xie$^\ast$\\Yubin Wang$^\ast$}
\email{{xiejiaqing, wangyubin}@pjlab.org.cn}
\affiliation{\institution{Shanghai Artificial Intelligence Laboratory}\city{Shanghai}\country{China}}

\author{Ben Gao\\Xingjian Wei\\Jiaxing Sun}
\email{{gaoben, weixingjian, sunjiaxing}@pjlab.org.cn}
\affiliation{\institution{Shanghai Artificial Intelligence Laboratory}\city{Shanghai}\country{China}}

\author{Tianfan Fu}
\email{futianfan@pjlab.org.cn}
\affiliation{\institution{Nanjing University}\city{Nanjing}\country{China}}
\affiliation{\institution{Shanghai Artificial Intelligence Laboratory}\city{Shanghai}\country{China}}

\author{Jiang Wu$^\dagger$\\Conghui He$^\dagger$}
\email{{wujiang, heconghui}@pjlab.org.cn}
\affiliation{\institution{Shanghai Artificial Intelligence Laboratory}\city{Shanghai}\country{China}}

\author{Yuqiang Li$^\dagger$\\Qinying Gu$^\dagger$}
\email{{liyuqiang, guqinying}@pjlab.org.cn}
\affiliation{\institution{Shanghai Artificial Intelligence Laboratory}\city{Shanghai}\country{China}}

\renewcommand{\shortauthors}{Yang et al.}

\begin{abstract}
Nuclear Magnetic Resonance (NMR) spectroscopy is fundamental for molecular structure elucidation, yet interpreting spectra at scale remains time-consuming and highly expertise-dependent. 
While recent spectrum-as-language modeling and retrieval-based methods have shown promise, they rely heavily on large corpora of computed spectra and exhibit notable performance drops when applied to experimental measurements.
 To address these issues, we build \textbf{\dataset}, a large-scale corpus of experimental $^1$H and $^{13}$C NMR spectra mined from chemical literature, 
 and propose \textbf{\mname}, which models spectra as unordered peak sets and aligns the model's inductive bias with the physical nature of NMR. 
 To our best knowledge, \mname\ is the first NMR Transformer trained solely on large-scale experimental spectra and achieves state-of-the-art performance on experimental benchmarks, improving Top-10 Accuracy over the strongest baseline by $+17.82$ points (61.15\% vs.\ 43.33\%), and underscoring the importance of experimental data and structure-aware architectures for reliable NMR structure elucidation.
\end{abstract}



\maketitle

\section{Introduction}

Nuclear Magnetic Resonance (NMR) spectroscopy stands as a cornerstone analytical technique in organic chemistry,
offering uniquely detailed insights into molecular structure at the atomic level~\cite{Silverstein1962, Silverstein2014}.
Unlike other spectroscopic methods such as infrared (IR)~\cite{Smith2018} and mass spectrometry (MS)~\cite{Boyd1997BookRM, Gross2017}, which primarily capture vibrational modes or molecular weight information, NMR reveals atomic connectivity, stereochemistry, and local electronic environments through chemical shifts and coupling constants~\cite{Duddeck1989, Kwan2008}.
This makes NMR indispensable for molecular structure elucidation, particularly in natural product discovery, pharmaceutical development, and synthetic chemistry~\cite{CLARIDGE201611, Roberts2000, Chiliveri2021, Huang2023}.

Despite its analytical power, interpreting NMR spectra remains a labor-intensive and expertise-dependent process that has not kept pace with modern high-throughput chemistry.
The forward problem of NMR spectroscopy is to obtain spectra from known molecular structures. It is well-defined and routinely achievable through direct experimental acquisition, quantum chemical simulation~\cite{Schreckenbach1995}, or machine learning-based prediction~\cite{Willcott2009}.
In contrast, the inverse problem aims to infer molecular structure from spectra and still largely relies on manual expert interpretation~\cite{Huang2021, yang2025spectrumworld}.
Conventional elucidation workflows involve iterative reasoning over chemical-shift assignments, coupling patterns, and functional group inference, rendering the process non-scalable and error-prone~\cite{Burns2019}.
As automated synthesis and closed-loop experimentation become increasingly prevalent in chemical research~\cite{Ley2015, Burger2020}, the bottleneck of manual spectral interpretation has emerged as a critical obstacle to autonomous molecular discovery~\cite{Huang2021}, motivating the need for reliable and scalable solutions to this inverse problem.

Building on these advances, a growing body of work has explored automated molecular structure elucidation from NMR spectra using data-driven approaches.
Most existing methods formulate the inverse problem as a conditional generation task, encoding NMR spectra and autoregressively decoding molecular representations such as SMILES \cite{nmr2struct, xue2025nmrmind, zhou2025nmrformer}.
Despite promising progress, these approaches suffer from two fundamental limitations.
First, due to the scarcity of curated experimental NMR data, many models are trained primarily on simulated spectra,  which leads to substantial performance degradation when applied to real experimental spectra~\cite{Mirza2024, nmrsolver}.
Second, NMR spectra are commonly treated as ordered sequences and encoded using standard Transformer with positional encodings, implicitly imposing artificial ordering among spectral peaks that lacks physical meaning and can mislead the model~\cite{Alberts2023, vita2024leveraging, nmr2struct}.

To address these limitations, we propose \textbf{\mname}, a framework for automated NMR-based structure elucidation that learns directly from large-scale experimental spectra and explicitly respects the structural properties of NMR spectral data. Our main contributions are summarized as follows:

\begin{itemize}[leftmargin=*, nosep]
    \item \textbf{Curation of large-scale experimental NMR dataset.} We construct \textbf{\dataset}, a large-scale corpus of experimental $^1$H and $^{13}$C NMR spectra mined from chemical literature, enabling model training on realistic experimental spectral distributions that are inaccessible through simulation alone.
    
    \item \textbf{Set Transformer architecture for NMR spectra.} We customize a \textit{Set Transformer}-based encoder \cite{lee2019set} to process $^1$H and $^{13}$C NMR spectra through parallel encoding pathways. By retaining permutation invariance and removing positional encodings, the proposed architecture focuses on chemically meaningful features, rather than spurious positional relationships among peaks.
    \item \textbf{Superior performance on experimental benchmarks.} Extensive experiments show that \mname\ consistently outperforms state-of-the-art generative and retrieval-based baselines in both accuracy and robustness on experimental NMR spectra, highlighting the importance of experimental data and architectures that respect the intrinsic structural properties of spectral data.
\end{itemize}
\section{Preliminary}
\label{sec:preliminary}

\begin{figure}[t]
    \centering
    \includegraphics[width=\linewidth]{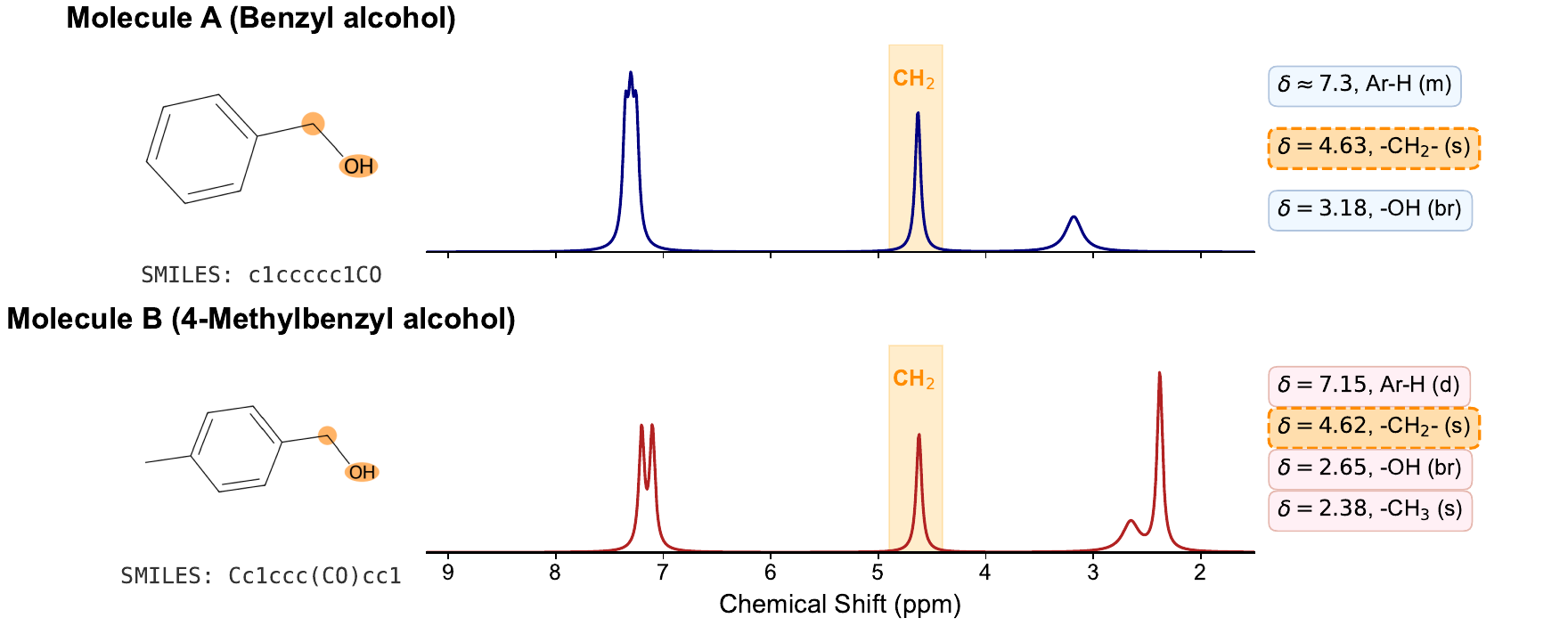}
    \caption{Local chemical environments determine NMR spectral features.}
    \label{fig:nmr_chemical_basis}
\end{figure}



\subsection{The Nature of NMR}
\label{subsec:nmr_as_sets}

As illustrated in Figure~\ref{fig:nmr_chemical_basis}, NMR spectral features are determined by local chemical environments. The ordering of peaks, however, is arbitrary and carries no structural information. Consequently, NMR spectra are naturally modeled as unordered peak sets.  Treating spectra as sequences therefore introduces an artificial ordering bias, which can cause predictions to depend on arbitrary indexing rather than chemical attributes. The task instead requires \emph{permutation invariance}. Let a spectrum be represented as a feature matrix $X \in \mathbb{R}^{n \times d}$, and let $\Pi$ denote any $n \times n$ permutation matrix that reorders peaks. A valid model $f$ should satisfy
\[
f(\Pi X) = f(X), \qquad \forall \Pi.
\]
This motivates architectures that operate directly on unordered sets and avoid sequence-specific inductive biases.

\subsection{Set Transformer Primitives}
\label{subsec:set_primitives}

To explicitly model the permutation-invariant nature of NMR data, we leverage the Set Transformer~\cite{lee2019set} framework. As shown in Figure~\ref{fig:set_transformer_blocks}, this architecture consists of attention-based blocks designed to process unordered sets.

\noindent\textbf{Multihead Attention Block (MAB).} 
MAB is the fundamental building block that generalizes standard attention to interactions between two sets. As depicted in Figure~\ref{fig:set_transformer_blocks}(a), it takes two input sets $X \in \mathbb{R}^{n \times d}$ and $Y \in \mathbb{R}^{m \times d}$ and performs attention:
\begin{gather}
H = \mathrm{LayerNorm}(X + \mathrm{MultiHeadAttn}(X, Y, Y)) \label{eq:mab_step1} \\
\mathrm{MAB}(X, Y) = \mathrm{LayerNorm}(H + \mathrm{FFN}(H)) \label{eq:mab_step2}
\end{gather}
To explicitly connect this to the standard attention formulation $\mathrm{Attention}(Q, K, V)$, MAB maps the input set $X$ to the Queries ($Q$) and the set $Y$ to both the Keys ($K$) and Values ($V$). This mechanism allows the elements in set $X$ to query and aggregate information from set $Y$.

\noindent\textbf{Set Attention Block (SAB).} 
SAB corresponds to the standard self-attention mechanism, defined as $\mathrm{SAB}(X) = \mathrm{MAB}(X, X)$. In this case, the mapping becomes $Q=K=V=X$, meaning the set attends to itself. While capturing pairwise interactions among all elements, SAB incurs a quadratic computational complexity of $\mathcal{O}(n^2)$, which limits scalability for high-resolution spectra.

\noindent\textbf{Induced Set Attention Block (ISAB).} 
To overcome the $\mathcal{O}(n^2)$ bottleneck, ISAB (Figure~\ref{fig:set_transformer_blocks}(c)) introduces a small set of $m$ learnable \textit{inducing points} $I$ (where $m \ll n$). Instead of computing attention directly between all input elements, ISAB uses $I$ as a low-rank approximation:
\begin{gather}
H = \mathrm{MAB}(I, X) \label{eq:isab_step1} \\
\mathrm{ISAB}_m(X) = \mathrm{MAB}(X, H) \label{eq:isab_step2}
\end{gather}
In the first step ($Q=I, K=V=X$), the inducing points summarize global features from the input. In the second step ($Q=X, K=V=H$), the input $X$ attends only to these summarized features. This mechanism effectively reduces the computational complexity from $\mathcal{O}(n^2)$ to $\mathcal{O}(nm)$.

\noindent\textbf{Pooling by Multihead Attention (PMA).} 
To aggregate a variable-sized set into a fixed-size representation, PMA uses a learnable seed matrix $S \in \mathbb{R}^{k \times d}$ as the query:
\[
\mathrm{PMA}_k(X) = \mathrm{MAB}(S, X).
\]
Since the query $S$ is fixed and independent of the input order, the pooled output is permutation invariant.

\noindent\textbf{{Summary.}}
In summary, Set Transformer constructs expressive set functions by stacking permutation-equivariant attention layers (MAB, SAB, or ISAB) followed by permutation-invariant pooling (PMA). These components form a flexible and theoretically grounded toolkit for modeling unordered sets, which we will adopt in the subsequent methodology section. Figure~\ref{fig:set_transformer_blocks} illustrates the core attention primitives used in this framework.
\section{Methodology}
\label{sec:methodology}

\begin{figure*}
\centering\includegraphics[width=0.85\linewidth]{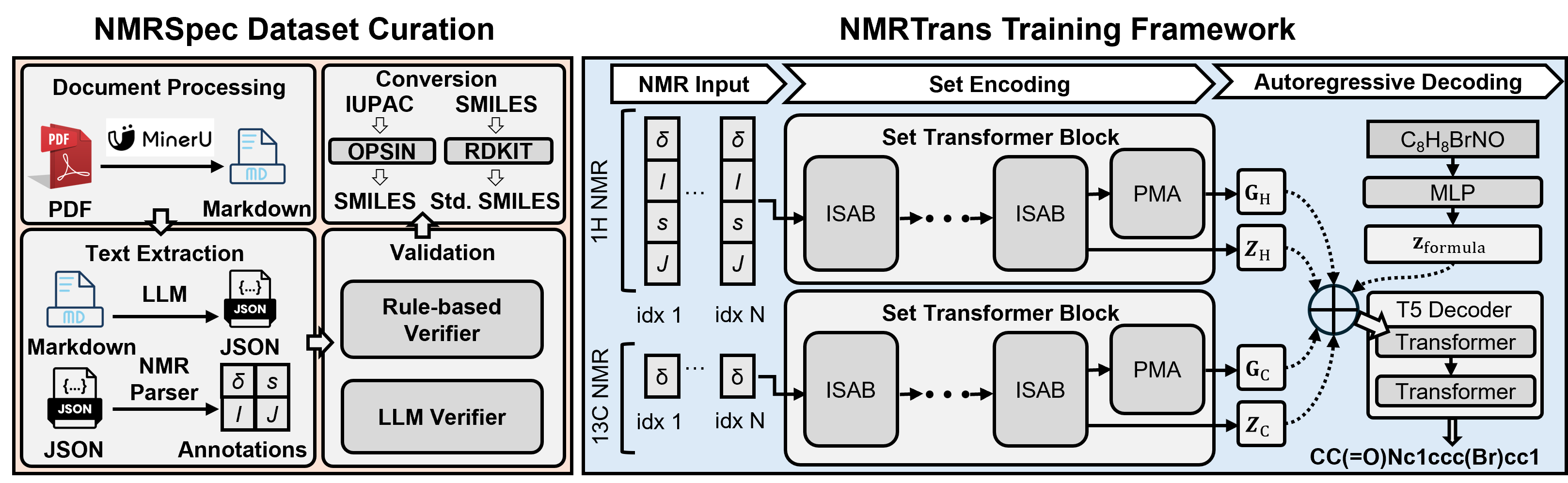}
    \caption{Left: Curation of our \dataset. Right: Pipeline of \mname: Set Transformer encoders for $^1$H/$^{13}$C peak sets, feature concatenation (optionally with molecular formula), and a T5 decoder for SMILES generation.}
    \label{fig:main}
\end{figure*}

\subsection{Problem Formulation}
\label{subsec:problem_formulation}

We consider the task of NMR-based molecular structure elucidation, where the goal is to infer a molecular structure from NMR spectra. Given one-dimensional peak sets from $^1$H and $^{13}$C NMR spectra, denoted as $\mathcal{P}_{\text{H}}$ and $\mathcal{P}_{\text{C}}$ respectively, optionally accompanied by a molecular formula $\mathcal{F}$, the model predicts a molecular structure represented as a canonical SMILES sequence $\mathbf{x} = (x_1, \ldots, x_L)$. This task can be formulated as learning the conditional distribution
\begin{equation}
    p(\mathbf{x} \mid \mathcal{P}_{\text{H}}, \mathcal{P}_{\text{C}}, \mathcal{F}).
\end{equation}



A key challenge of this inverse problem lies in the representation of NMR spectra. In contrast to natural language, NMR spectra do not possess meaningful sequential structure: the index of peaks is arbitrary. Instead, chemical information is encoded in the physical coordinate($\delta$) and other peak attributes, such as splitting patterns and coupling statistics, rather than by the sequential arrangement of input elements. 

Consequently, the end-to-end mapping must be \textit{permutation invariant} with respect to peak ordering.
Let $X_{\text{H}}$ and $X_{\text{C}}$ denote feature matrices of the $^1$H and $^{13}$C peak sets, and let $\Pi_{\text{H}}$ and $\Pi_{\text{C}}$ be permutation matrices that reorder peaks.
The model $f_{\theta}$ should satisfy
\begin{equation}
    f_{\theta}(\Pi_{\text{H}} X_{\text{H}}, \Pi_{\text{C}} X_{\text{C}}, \mathcal{F}) = f_{\theta}(X_{\text{H}}, X_{\text{C}}, \mathcal{F}), \quad \forall \Pi_{\text{H}}, \Pi_{\text{C}},
\end{equation}
ensuring identical molecular predictions for spectra differing only in peak ordering. This invariance is achieved through a composition of \textit{permutation-equivariant} intermediate layers (e.g., ISAB) followed by \textit{permutation-invariant} aggregation operations (e.g., PMA). Our architecture enforces this property at the architectural level rather than relying on data augmentation.

\subsection{Spectrum-aware Feature Engineering}
\label{subsec:feature_engineering}

We represent NMR spectra as collections of peaks with physically meaningful acquisition attributes. For a $^1$H NMR spectrum, each peak $p_i^{\text{H}}$ is characterized by
\begin{equation}
    p_i^{\text{H}} =
    \big(
        \delta_i,\ 
        I_i,\ 
        s_i,\ 
        \mathbf{j}_i
    \big),
    \label{eq:hnmr_peak}
\end{equation}
where $\delta_i \in [0,12]$ ppm denotes the chemical shift (defined as the centroid of the multiplet), $I_i$ the integration intensity, $s_i$ the splitting pattern (e.g., doublet, triplet). The feature vector $\mathbf{j}_i = [J_{i,1}, \dots, J_{i,k}] \in \mathbb{R}^k$ contains the associated $J$-coupling constants defining the splitting structure.
These constants are encoded via sorting followed by zero-padding to a fixed length (e.g., $k=6$), capturing the local proton connectivity environment.

For $^{13}$C NMR, due to broadband decoupling, each signal $p_j^{\text{C}}$ appears as a singlet and is described solely by its chemical shift:
\begin{equation}
    p_j^{\text{C}} = \delta_j, \quad \delta_j \in [0,220]\,\text{ppm}
    \label{eq:cnmr_peak}
\end{equation}
as integration values are not quantitatively reliable under broadband decoupling conditions~\cite{claridge2016high}. The wide dispersion of $^{13}$C chemical shifts provides complementary carbon-type information that is largely absent from $^1$H spectra.

Accordingly, the input to our model consists of two unordered sets of signal features:
\begin{equation}
    \mathcal{P}_{\text{H}} = \{p_1^{\text{H}}, \dots, p_{N_{\text{H}}}^{\text{H}}\}, \quad
    \mathcal{P}_{\text{C}} = \{p_1^{\text{C}}, \dots, p_{N_{\text{C}}}^{\text{C}}\},
    \label{eq:peak_sets}
\end{equation}
where $N_{\text{H}}$ and $N_{\text{C}}$ denote the number of peaks in the $^1$H and $^{13}$C spectra, respectively.
These sets are embedded into continuous feature matrices $X_{\text{H}} \in \mathbb{R}^{N_{\text{H}} \times d}$ and $X_{\text{C}} \in \mathbb{R}^{N_{\text{C}} \times d}$ to serve as inputs for the Set Transformer encoders.

\subsection{Encoder: Hierarchical Representation for NMR Spectra}
\label{subsec:set_transformer}

To model NMR spectra in a manner consistent with their physical properties, we design a permutation-equivariant encoder inspired by the Set Transformer architecture~\cite{lee2019set}. The key requirement is that the encoder operates on the spectral feature matrices $X_{\text{H}}$ and $X_{\text{C}}$ based solely on their chemical attributes, while remaining invariant to any arbitrary ordering.

\noindent\textbf{Permutation-equivariant encoding.}
As established in Section~\ref{subsec:problem_formulation}, intermediate spectral representations must satisfy permutation equivariance ($f(\Pi X) = \Pi f(X)$) to preserve peak-feature correspondence, while the final pipeline output must be permutation invariant. Unlike standard Transformers that rely on positional encodings, our architecture operates directly on the unordered feature matrix $X \in \{X_{\text{H}}, X_{\text{C}}\}$. By processing peaks as a set rather than a sequence, the encoder inherently guarantees equivariant intermediate representations and, through subsequent invariant aggregation, invariant final predictions.

\noindent\textbf{Induced Set Attention as inductive bias.}
Beyond permutation equivariance, effective spectral modeling requires controlling the interaction structure among peaks. Vanilla self-attention computes pairwise interactions across all $N$ peaks, which becomes increasingly diffuse as $N$ grows and as chemically irrelevant or noisy peaks are introduced. To address this, we employ \emph{Induced Set Attention Blocks} (ISAB), which introduce a learnable bottleneck that constrains information flow through a small set of inducing points.

Formally, following the definition in Eq.~\ref{eq:isab_step1} and Eq.~\ref{eq:isab_step2}, an ISAB layer consists of two consecutive multi-head attention operations:
\begin{equation}
H = \mathrm{MAB}(I, X), \qquad Z = \mathrm{MAB}(X, H),
\label{eq:isab}
\end{equation}
where $X \in \mathbb{R}^{N \times d}$ denotes the input peak embeddings, and $I \in \mathbb{R}^{m \times d}$ denotes $m \ll N$ learnable inducing points. This two-stage structure encodes a strong inductive bias:
\begin{itemize}[leftmargin=*]
    \item In the first stage ($I \rightarrow X$), inducing points attend to the entire spectrum, compressing the peak set into a small number of latent spectral concepts.
    \item In the second stage ($X \rightarrow H$), individual peaks attend to these latent concepts, refining their representations while suppressing spurious peak--peak interactions.
\end{itemize}
By mediating interactions through $m$ inducing points, ISAB enforces a hierarchical representation that may facilitate coarse-to-fine reasoning during spectral interpretation.

\noindent\textbf{Set Transformer block architecture.}
For each NMR modality, we process the corresponding feature matrix through a Set Transformer block composed of $L$ stacked ISAB layers, followed by Pooling by Multihead Attention (PMA):
\begin{equation}
Z = \mathrm{ISAB}_L(X), \qquad G = \mathrm{PMA}_k(Z).
\label{eq:set_transformer_block}
\end{equation}
Here, $Z \in \mathbb{R}^{N \times d}$ encodes refined peak-level features (corresponding to $Z_{\text{H}}$ or $Z_{\text{C}}$ in Eq.~\ref{eq:fusion}), while $G \in \mathbb{R}^{k \times d}$ captures global spectral context (corresponding to $G_{\text{H}}$ or $G_{\text{C}}$). We set $k=4$ to obtain a multi-faceted global representation. ISAB is permutation-equivariant by construction, while PMA is permutation-invariant with respect to its input set for any $k$.

Although the peak-level representation $Z$ remains permutation-equivariant, the subsequent cross-attention in the decoder aggregates these features into a permutation-invariant context vector (Lemma 3, Appendix~\ref{app:permutation}), ensuring strict order independence for final molecular predictions. This composition of equivariant feature extraction followed by invariant aggregation forms the theoretical foundation of our architecture. Unless otherwise specified, we set the number of inducing points to $m=32$ and stack $L=6$ ISAB layers, which we found to provide a favorable trade-off between abstraction capacity and retention of chemically relevant detail.

\subsection{Multi-Modal Fusion}
\label{subsec:fusion}

Each input modality is encoded independently using its corresponding spectral encoder as described in Section~\ref{subsec:set_transformer}. Specifically, for the feature matrices $X_{\text{H}}$ and $X_{\text{C}}$, we obtain their respective peak-level and global representations:
\[
(Z_{\text{H}}, G_{\text{H}}) = f_{\text{H}}(X_{\text{H}}), \qquad
(Z_{\text{C}}, G_{\text{C}}) = f_{\text{C}}(X_{\text{C}}).
\]

When available, the molecular formula $\mathcal{F}$ (represented as a string, e.g., ``C10H16O2'') is tokenized into characters, embedded via a lookup table, and encoded by a multilayer perceptron into a global constraint vector $G_{\text{F}} \in \mathbb{R}^{1 \times d}$. Importantly, each modality ($^1$H NMR, $^{13}$C NMR, and molecular formula) is optional and can be independently masked, allowing the model to operate flexibly under different experimental settings.

To integrate both local (peak-level) and global (spectrum-level) cues, the final encoder output is formed by concatenating all available representations along the sequence dimension:
\begin{equation}
H_{\text{enc}} = \left[ G_{\text{C}}, Z_{\text{C}}, G_{\text{H}}, Z_{\text{H}}, G_{\text{F}} \right],
\label{eq:fusion}
\end{equation}
where representations corresponding to unavailable modalities are omitted. The concatenation order is arbitrary and does not affect final predictions due to the permutation invariance of cross-attention (Lemma 3). By treating each global vector in $G$ as an independent set element, this fusion strategy preserves complementary spectral information while maintaining permutation invariance. Critically, no positional encodings are injected at any stage—from spectral input through fusion to decoder cross-attention—ensuring strict adherence to the unordered nature of NMR peak sets. The fused representation $H_{\text{enc}}$ is subsequently consumed by the decoder for molecular structure generation.

\subsection{Autoregressive Decoder}
\label{subsec:decoder}

The decoder generates the molecular SMILES sequence autoregressively conditioned on the fused encoder representation $H_{\text{enc}}$.
Following~\cite{vita2024leveraging}, we adopt the T5 architecture~\cite{raffel2020exploring} as the backbone, while introducing targeted modifications to accommodate the unordered nature of NMR spectral inputs.

\noindent\textbf{SMILES Vocabulary.}
Molecular structures are represented using canonical SMILES strings and decoded from a chemically informed output vocabulary.
We employ a custom regular expression-based tokenizer to define atomic-level tokens that preserve chemically meaningful units.
Specifically, the vocabulary includes multi-character atomic symbols (e.g., \texttt{Cl}, \texttt{Br}), bracketed stereochemical expressions (e.g., \texttt{[C@H]}, \texttt{[N+]}), bond types, ring-closure digits, and branch delimiters.
This design avoids the fragmentation of chemical entities inherent to character-level tokenization (e.g., preventing \texttt{Cl} from being split into \texttt{C} and \texttt{l}), while ensuring exact reversibility between token sequences and canonical SMILES strings.
In addition to chemical tokens, the vocabulary contains special symbols \texttt{[BOS]}, \texttt{[EOS]}, and \texttt{[PAD]}.

\noindent\textbf{Training Objective.}
Given the encoder representation $H_{\text{enc}}$, the decoder models the conditional probability of a target SMILES sequence $\mathbf{x}$ in an autoregressive manner:
\begin{equation}
    p(\mathbf{x} \mid H_{\text{enc}}) = \prod_{t=1}^{L} p(x_t \mid \mathbf{x}_{<t}, \text{CrossAttn}(q_t, H_{\text{enc}})).
\end{equation}
Training minimizes the cross-entropy loss between predicted token distributions and the ground-truth sequence $\mathbf{x}^*$:
\begin{equation}
    \mathcal{L} = -\sum_{t=1}^{L} \log p(x_t^* \mid \mathbf{x}_{<t}^*, H_{\text{enc}}),
\end{equation}
where the dependency on the spectral information is mediated through the cross-attention mechanism.

\noindent\textbf{Removal of Positional Bias.}
In standard T5 decoders, relative positional biases are introduced in attention layers, implicitly encoding order information in the encoder states.
However, for set-structured inputs such as NMR peaks, any imposed ordering is arbitrary and may introduce spurious dependencies.
To eliminate this effect, we explicitly \textbf{remove all positional biases} from the cross-attention modules.
As a result, attention weights depend solely on the content compatibility between the decoding state and the chemical attributes encoded in $H_{\text{enc}}$ (e.g., chemical shift values), rather than on the input indices.

\noindent\textbf{Permutation Invariance via Cross-Attention.}
As a direct consequence of removing positional biases, the decoder achieves permutation invariance with respect to the encoder states.
Formally, for a decoder query $q_t$ and encoder states $H_{\text{enc}} = \{\mathbf{h}_j\}_{j=1}^M$ (where $M$ is the total length of the fused sequence), cross-attention is computed as
\begin{equation}
    \text{CrossAttn}(q_t, H_{\text{enc}}) = \sum_{j=1}^{M} \alpha_{tj} \mathbf{h}_j, 
    \quad \alpha_{tj} = \frac{\exp(q_t^\top \mathbf{h}_j / \sqrt{d})}{\sum_{k=1}^{M} \exp(q_t^\top \mathbf{h}_k / \sqrt{d})}.
    \label{eq:cross_attn}
\end{equation}
If the encoder states are permuted, the attention weights are recomputed accordingly and remain paired with their corresponding values.
Combined with the commutativity of summation in the softmax denominator, this guarantees that the resulting weighted sum is invariant to the input ordering.
A rigorous proof of this property is provided in Lemma~3 (Appendix~\ref{app:permutation}).
\section{Dataset}
\label{sec:dataset}

In this paper, we present \textbf{\dataset}, a large-scale corpus of \emph{experimental} molecular spectroscopy records mined from chemical journals. The curation pipeline is illustrated in the 'NMRSpec Dataset Curation' block in Figure~\ref{fig:main}.

\noindent\textbf{Motivation: Bridging the Simulation--Experiment Gap.} 
The development of robust, generalizable AI models for NMR structure elucidation critically depends on access to large-scale, high-quality experimental data. However, publicly available experimental NMR repositories remain scarce. For instance, widely used databases like NMRShiftDB2~\cite{Steinbeck2003} contain only $\sim$50k entries, falling far short of the data scale required for training modern deep neural networks~\cite{Wang2025}. 
Consequently, prior approaches have largely relied on computationally simulated spectra. While scalable, simulated spectra inevitably deviate from experimental measurements due to complex solvent effects, impurities, and magnetic interactions. This discrepancy creates a substantial \emph{simulation--experiment domain gap}, leading to severe performance degradation when models trained on synthetic data are applied to real-world tasks. 
Notably, Kluger et al.~\cite{nmr2struct} reported that the accuracy of their generative framework dropped precipitously from 69.6\% to 33.0\% when evaluated on a set of 106 experimental $^1$H and $^{13}$C spectra compared to simulated benchmarks. 
To address this bottleneck and enable reliable real-world application, we introduce \textbf{\dataset} to provide the necessary large-scale experimental supervision that prior synthetic datasets lack.

\noindent\textbf{Data Curation Pipeline.} 
We developed a systematic five-stage pipeline to harvest structured molecular and spectroscopic data from the Supporting Information (SI) of peer-reviewed chemistry journals:
(1) \textbf{Document Parsing:} Raw PDF documents were processed using \texttt{MinerU} to extract textual content into Markdown format while preserving high-resolution images.
(2) \textbf{Spectroscopic Extraction:} We employed Large Language Models (LLMs) to identify spectroscopy-related sections, followed by a rule-based parser using regular expressions to precisely extract NMR parameters (e.g., chemical shifts, $J$-couplings, multiplicities).
(3) \textbf{Structure Standardization:} Chemical names (IUPAC) were converted to SMILES strings using OPSIN~\cite{Terakita2005} and canonicalized via RDKit to ensure unique representations.
(4) \textbf{Image Analysis:} For documents containing spectral images, vision-language models were utilized to extract metadata and link images to their corresponding textual records.
(5) \textbf{Validation:} Rigorous deduplication and chemical validity checks were applied to ensure data integrity.

\noindent\textbf{Corpus Statistics.} 
The resulting raw corpus, which we refer to as the \textbf{\dataset Corpus}, covers contemporary chemistry literature from \textbf{2013 to 2025}. 
Processing \textbf{62,071} documents yielded a massive collection of \textbf{2,144,492} spectroscopic records across \textbf{681,990} unique molecules. 
Beyond the scale, a key contribution is the diversity of analytical techniques: while $^1$H NMR ($\sim$29.2\%) and $^{13}$C NMR ($\sim$24.5\%) constitute the majority, the corpus also encompasses Mass Spectrometry ($\sim$25.2\%) and IR ($\sim$11.7\%), providing a rich resource for future multi-modal research. 
Comprehensive statistics are detailed in \textbf{Appendix~\ref{app:dataset_details}}.

\noindent\textbf{Benchmark Construction.} 
From this broad corpus, we constructed a high-quality benchmark tailored specifically for the structure elucidation task.
We applied a strict filtering protocol to select organic molecules (containing B, Br, C, Cl, F, H, I, N, O, P, S, or Si) that possess \textbf{paired} experimental $^1$H and $^{13}$C NMR spectra, as our method relies on the fusion of both modalities.
We further filtered out overly complex samples (e.g., $>60$ peaks or SMILES length $>80$) to ensure training stability. 
Finally, we adopted a random splitting strategy to partition this filtered dataset into training, validation, and test sets with a ratio of 8:1:1. This resulted in a final benchmark of approximately \textbf{240,000 samples} (190k training, 25k validation, and 25k test), which serves as the experimental basis for all subsequent evaluations.
\section{Experiments}
\label{sec:experiment}

\subsection{Experimental Setup}
\label{sec:experiment_setup}

\noindent\textbf{Baselines.}
We evaluate \mname\ against representative baselines spanning dominant paradigms for NMR-based structure elucidation.
Accordingly, we select two generative transformer-based models: \textbf{NMR2Struct}~\cite{nmr2struct} and \textbf{NMRMind}~\cite{xue2025nmrmind}.
Both methods \textbf{discretize} spectra by encoding chemical shift and signal intensities into token sequences, allowing transformers to attend to discrete spectral inputs for direct SMILES decoding.
Distinguished from these generative approaches, \textbf{NMR-Solver}~\cite{nmrsolver} represents the retrieval-and-optimization paradigm.
It adopts a database-driven strategy that performs nearest-neighbor matching against a simulated spectral library, followed by fragment-level optimization to \textbf{refine} the structure.

\noindent\textbf{Evaluation metrics.}
We evaluate structure elucidation performance using three complementary metrics.
\textbf{Sequence-level accuracy} measures whether the predicted \textbf{canonical} SMILES exactly matches the reference structure, reflecting end-to-end correctness.
\textbf{Token-level accuracy} computes the fraction of correctly predicted SMILES tokens, capturing partial structural agreement.
In addition, we report \textbf{Tanimoto similarity} between predicted and ground-truth molecules based on RDKit Morgan fingerprints (radius=2, nBits=2048, without chirality) to quantify structural similarity when exact matches are not achieved.
All metrics are reported at Top-$k$ ($k \in \{1,5,10\}$).

\noindent\textbf{Implementation details.}
All methods are evaluated on the same experimental NMR test set.
For all generative models, including ours and transformer-based baselines, we employ beam search with a beam size of 10 to produce Top-$k$ candidate SMILES.
To ensure fair comparison under identical data distributions, all baseline models are trained from scratch on our experimental training set using their official implementations, including spectral encoding schemes, molecular formula embeddings, and tokenizer vocabularies.
For the retrieval-based baseline NMR-Solver, we follow its three-stage protocol: initially retrieving the top-100 candidates via FAISS vector similarity, re-ranking them based on explicit $^1$H and $^{13}$C spectral matching scores, and finally filtering candidates based on molecular formula and structural validity to obtain the final top-$k$ predictions.

\begin{table*}[!t]
\centering
\setlength{\tabcolsep}{3.5pt}
\caption{
Quantitative comparison on experimental $^1$H and $^{13}$C NMR spectra from \dataset.
Best results are \textbf{bolded}. 
Results report mean$_{\pm\text{std}}$ over three independent runs; 
NMR-Solver is deterministic (std = 0.00).
}
\label{tab:main_results}
\begin{tabular}{@{}lccccccccc@{}}
\toprule
\multirow{2}{*}{Method}
& \multicolumn{3}{c}{Top-1}
& \multicolumn{3}{c}{Top-5}
& \multicolumn{3}{c}{Top-10} \\
\cmidrule(lr){2-4} \cmidrule(lr){5-7} \cmidrule(lr){8-10}
& Seq Acc & Tok Acc & Tanimoto
& Seq Acc & Tok Acc & Tanimoto
& Seq Acc & Tok Acc & Tanimoto \\
\midrule
NMR2Struct~\cite{nmr2struct}
& $3.13_{\pm0.02}$  & $30.54_{\pm0.04}$ & $27.30_{\pm0.09}$
& $4.76_{\pm0.03}$  & $41.18_{\pm0.05}$ & $38.97_{\pm0.04}$
& $5.53_{\pm0.04}$  & $45.69_{\pm0.05}$ & $42.41_{\pm0.01}$ \\
NMR-Solver (search-only)~\cite{nmrsolver}
& $22.35_{\pm0.00}$ & $29.22_{\pm0.00}$ & $34.72_{\pm0.00}$
& $23.81_{\pm0.00}$ & $30.16_{\pm0.00}$ & $35.23_{\pm0.00}$
& $23.83_{\pm0.00}$ & $30.17_{\pm0.00}$ & $35.26_{\pm0.00}$ \\
NMRMind~\cite{xue2025nmrmind}
& $37.33_{\pm0.01}$ & $63.51_{\pm0.03}$ & $72.71_{\pm0.04}$
& $41.96_{\pm0.12}$ & $68.33_{\pm0.05}$ & $77.05_{\pm0.03}$
& $43.33_{\pm0.03}$ & $69.71_{\pm0.02}$ & $78.14_{\pm0.00}$ \\
\mname{} (Ours)
& $\mathbf{42.81_{\pm0.05}}$ & $\mathbf{71.30_{\pm0.04}}$ & $\mathbf{74.62_{\pm0.10}}$
& $\mathbf{58.22_{\pm0.04}}$ & $\mathbf{81.46_{\pm0.01}}$ & $\mathbf{82.52_{\pm0.02}}$
& $\mathbf{61.15_{\pm0.04}}$ & $\mathbf{84.18_{\pm0.02}}$ & $\mathbf{84.41_{\pm0.03}}$ \\
\bottomrule
\end{tabular}
\end{table*}

\subsection{Overall Performance} 
\label{sec:overall_performance}

Table~\ref{tab:main_results} compares \mname{} against recent methods on experimental $^1$H and $^{13}$C NMR spectra from \dataset.
\mname{} achieves the highest Top-1 sequence accuracy (42.81\%), outperforming the strongest baseline NMRMind (37.33\%) by 5.48\%.
The performance gap widens under relaxed evaluation: at Top-5 and Top-10, \mname{} exceeds NMRMind by 16.26\% and 17.82\% in sequence accuracy, respectively (58.22\% vs.\ 41.96\% and 61.15\% vs.\ 43.33\%).

\mname{} also achieves the highest token accuracy and Tanimoto similarity across all Top-$k$ settings.
The absolute improvement in Tanimoto similarity increases from 1.91 percentage points at Top-1 to 6.27 percentage points at Top-10, indicating progressively better topological similarity in the candidate lists.
These results demonstrate that \mname{} consistently outperforms existing methods.

\subsection{Detailed Analysis}
\label{sec:detailed_analysis}

We conduct a series of fine-grained analyses to evaluate the robustness of \mname{} against baseline NMRMind across different evaluation regimes. Figure~\ref{fig:visualization_results} summarizes these comparative results.

\begin{figure}[t]
    \centering
    \begin{subfigure}{\linewidth}
        \small
        \centering
        \includegraphics[width=0.7\linewidth]{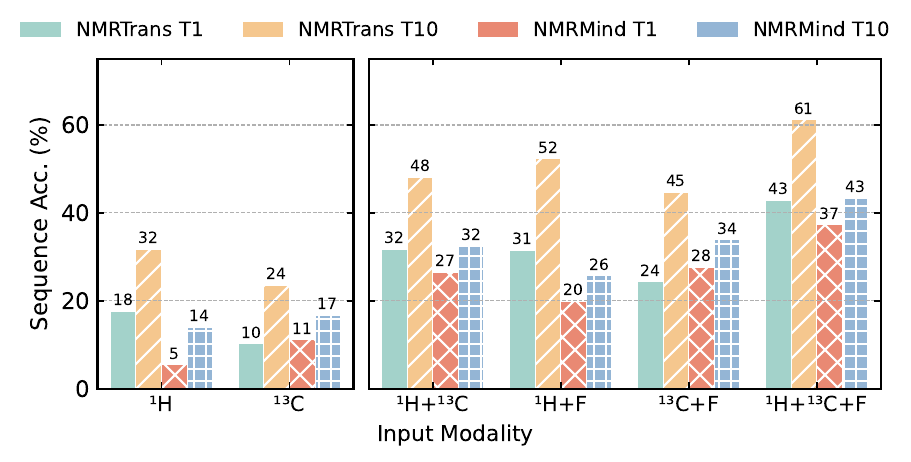}
        \caption{Input Modality Ablation}
        \label{fig:modality}
    \end{subfigure}
    
    \vspace{0.3em} 
    
    \begin{subfigure}{0.48\linewidth}
        \centering
        \includegraphics[width=\linewidth]{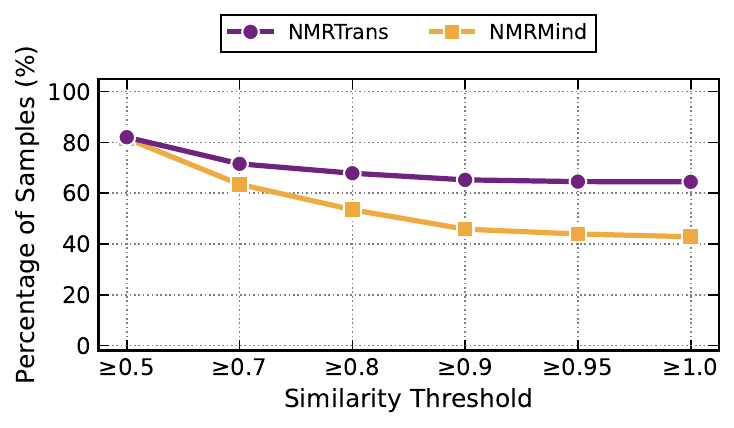}
        \caption{Similarity Thresholds}
        \label{fig:tani_threshold}
    \end{subfigure}
    \hfill 
    \begin{subfigure}{0.48\linewidth}
        \centering
        \includegraphics[width=\linewidth]{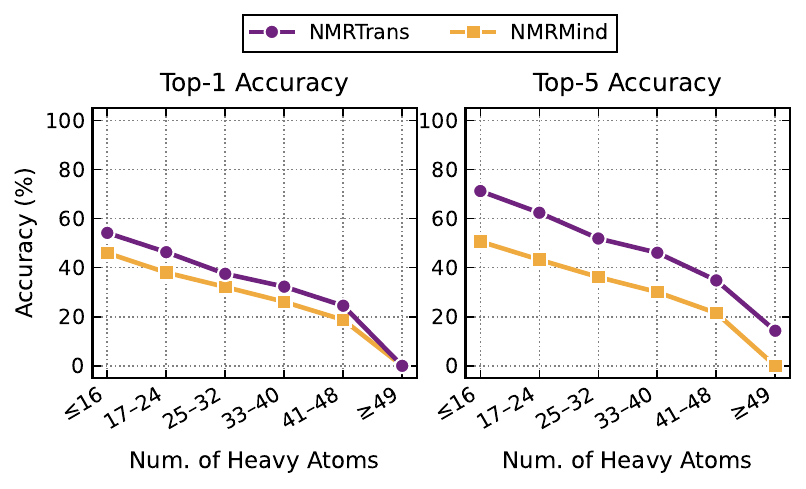}
        \caption{Effect of Molecular Size}
        \label{fig:mol_size}
    \end{subfigure}
    
    \caption{Detailed performance analysis of \mname{} with NMRMind. 
    (a) Top-1 and Top-10 Sequence accuracy comparison under varying input modality combinations (e.g., $^{1}$H, $^{13}$C, Formula).
    (b) Percentage of test samples where the Top-5 predictions meet specific Tanimoto similarity thresholds (x-axis: similarity score $\ge 0.5, 0.7, \dots, 1.0$). 
    (c) Prediction accuracy (Top-1 and Top-5) as a function of molecular complexity (number of heavy atoms).}
    \label{fig:visualization_results}
\end{figure}

\noindent\textbf{Robustness across input modalities.}
Figure~\ref{fig:modality} compares sequence-level accuracy across six input settings.
\mname{} demonstrates significant superiority in all proton-involved scenarios, notably achieving over $3\times$ the Top-1 accuracy of NMRMind in the $^1$H-only setting (18\% vs. 5\%).
In carbon-dominant settings ($^{13}$C, $^{13}$C+F), although the baseline shows slightly higher Top-1 accuracy, \mname{} consistently delivers higher Top-10 recall (e.g., 45\% vs. 34\% for $^{13}$C+F), indicating robust retrieval capability despite signal sparsity.
Crucially, under the full input setting ($^1$H+$^{13}$C+F), \mname{} effectively synergizes all modalities to achieve the best overall performance, surpassing the baseline by 6\% in Top-1 accuracy.

\noindent\textbf{Structural similarity distribution.}
Figure~\ref{fig:tani_threshold} evaluates the structural quality of the generated candidates.
For each test sample, we compute the \textbf{maximum Tanimoto similarity} between the ground truth and the Top-5 predictions, and report the percentage of samples where this maximum similarity exceeds specific thresholds ($\tau$).
While both models exhibit comparable coverage at lower thresholds ($\tau \geq 0.5$), the performance gap widens significantly as the criterion becomes stricter.
Most notably, at the exact reconstruction threshold ($\tau = 1.0$), \mname\ achieves a decisive \textbf{21.6\%} absolute improvement over NMRMind(64.48\% vs. 42.84\%).
This indicates that \mname\ significantly increases the likelihood of recovering the precise ground-truth molecule among the top candidates, rather than merely generating structurally plausible isomers.

\noindent\textbf{Effect of molecular complexity.}
Figure~\ref{fig:visualization_results}c reports the prediction accuracy as a function of molecular size (measured by the number of heavy atoms).
As expected, performance for both models degrades as molecules become larger ($>40$ atoms), reflecting the severe spectral congestion and ambiguity in such complex systems.

However, \mname\ exhibits significantly better scalability.
While the baseline's performance drops precipitously for larger molecules, \mname\ maintains a robust lead across all size bins.
Crucially, in the most challenging regime ($\ge 49$ heavy atoms), whereas NMRMind fails completely (dropping to \textbf{0\%} Top-5 accuracy), \mname\ retains meaningful predictive capacity ($\sim$15\%).
This demonstrates that our permutation-invariant spectral modeling is not only more accurate on average but also far more robust to the complexity of real-world macromolecular structures.

\subsection{Case Study}
\label{sec:case_study}

Figure~\ref{fig:case_study} highlights \mname{}'s robustness in resolving traditionally challenging structural motifs. 
Rows 1--4 demonstrate accurate reconstructions (Top-1 Tanimoto=1.0) across:
(1) \textbf{Long aliphatic chains} (Row 1), distinguishing specific chain lengths despite severe spectral overlap in the $0.5$--$2.0$~ppm region; 
(2) \textbf{Hetero-/Polycyclic systems} (Rows 2--3), capturing subtle chemical shift variations driven by heteroatom placement; 
and (3) \textbf{Heavier molecules} ($\ge 40$ atoms, Row 4), maintaining structural fidelity despite increased spectral congestion.

Row 5 illustrates the model's \textbf{recall capability} in ambiguous scenarios: while the highest-ranked prediction (Top-1) is incorrect, the ground truth is successfully captured within the Top-3 candidates, offering valuable options for expert verification.

\begin{figure}[htbp!]
\centering
\includegraphics[width=\linewidth]{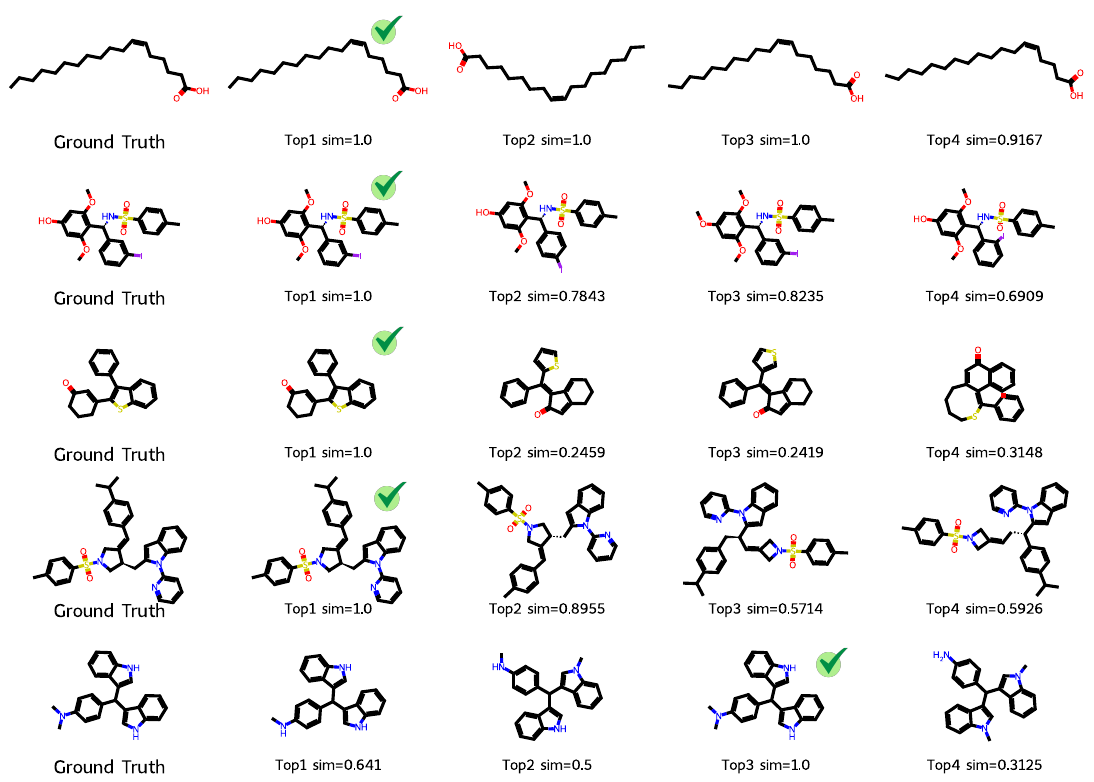}
\caption{\mname{} successfully reconstructs challenging motifs including long aliphatic chains, heterocycles, and heavy molecules ($\geq$40 atoms), with ground-truth structures recovered within the Top-3 predictions for ambiguous cases.}

\label{fig:case_study}
\end{figure}

\subsection{Ablation Study}
\label{sec:ablation}

\begin{table}[htbp]
    \centering
    \small
    \setlength{\tabcolsep}{4pt}
    \caption{
        Ablation study on spectral modalities and molecular formula constraints.
        \textbf{Best} and \underline{second-best} results are highlighted in \textbf{bold} and \underline{underline}, respectively.
        All metrics are reported in percentage (\%).
    }
    \label{tab:ablation_modalities}
    \begin{tabular}{@{}l c ccc cc@{}}
        \toprule
        \multicolumn{1}{c}{Condition} & \multicolumn{1}{c}{Formula}
        & \multicolumn{3}{c}{Seq Acc} 
        & Tanimoto & Token Acc \\
        \cmidrule(lr){3-7}
        & & Top-1 & Top-5 & Top-10 & Top-1 & Top-1 \\
        \midrule
        $^{13}$C NMR & \xmark
        & 10.19 & 20.71 & 23.53 & 42.06 & 44.05 \\
        $^{1}$H NMR & \xmark
        & 17.58 & 28.46 & 31.70 & 53.16 & 54.33 \\
        $^{1}$H+$^{13}$C NMR & \xmark
        & 31.68 & 44.78 & 48.16 & 67.08 & \underline{65.37} \\
        \midrule
        $^{13}$C NMR & \cmark
        & 24.30 & 41.03 & 44.67 & 59.31 & 57.38 \\
        $^{1}$H NMR & \cmark
        & \underline{34.89} & \underline{50.09} & \underline{53.47} & \underline{67.32} & 65.29 \\
        $^{1}$H+$^{13}$C NMR & \cmark
        & \textbf{42.81} & \textbf{58.22} & \textbf{61.15} & \textbf{74.62} & \textbf{71.30} \\
        \bottomrule
    \end{tabular}
\end{table}

\noindent\textbf{Effect of input modalities on structure recovery.}
Table~\ref{tab:ablation_modalities} quantifies the impact of spectral modalities and molecular formula constraints.
The full model ($^1$H+$^{13}$C+Formula) achieves the highest Top-1 accuracy (42.81\%), outperforming single-modality baselines by 11.13--32.62\%.
While single-spectrum models are limited by signal sparsity ($^{13}$C: 10.19\%) or lack of backbone context ($^1$H: 17.58\%), combining them boosts accuracy to 31.68\%, confirming that proton connectivity and carbon frameworks provide essential, complementary structural evidence.

Crucially, integrating molecular formula constraints yields consistent absolute gains of \textbf{11.13--17.31}\% across all settings.
By effectively pruning constitutionally impossible isomers, this global constraint significantly reduces the search space, validating the necessity of fusing local spectral features with global compositional priors for accurate elucidation.

\begin{table}[htbp]
\centering
\small
\setlength{\tabcolsep}{4pt}
\caption{
Ablation study on architectural components for NMR peak encoding.
VT: Vanilla Transformer; PE: positional encoding; MAB: Multihead Attention Block; 
PMA: Pooling by Multihead Attention \cite{lee2019set}.
All metrics are reported in percentage (\%).
}
\label{tab:ablation_on_architecture}
\begin{tabular}{@{}lccccc@{}}
\toprule
Model 
& \multicolumn{3}{c}{Seq Acc} 
& Tanimoto & Token Acc \\
\cmidrule(lr){2-6}
& Top-1 & Top-5 & Top-10 & Top-1 & Top-1 \\
\midrule
VT (w/ PE)
& 27.28 & 45.25 & 49.09 & 64.36 & 61.09 \\
VT (w/o PE)
& 33.36 & 48.69 & 52.20 & 68.87 & 65.13 \\
\midrule
\mname{} (w/o PMA)
& \underline{41.05} & \underline{57.28} & \underline{60.50} & \underline{73.84} & \underline{70.08} \\
\mname{} (w/ PMA)
& \textbf{42.81} & \textbf{58.22} & \textbf{61.15} & \textbf{74.62} & \textbf{71.30} \\
\bottomrule
\end{tabular}
\vspace{-3em}
\end{table}

\noindent\textbf{Architectural inductive biases for unordered spectral data.}
\begin{figure}[!htb]
    \centering
    \includegraphics[width=.6\linewidth]{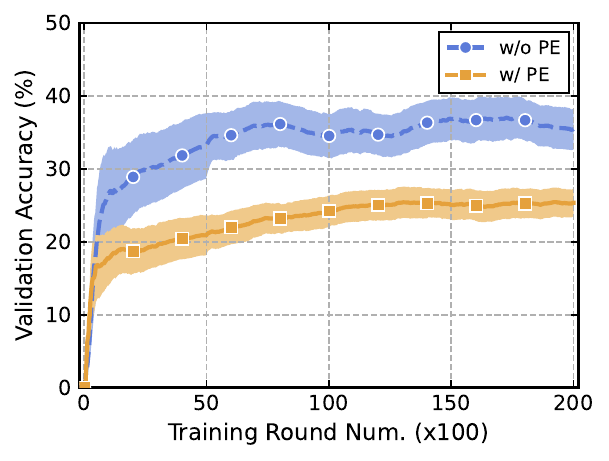}
    \caption{Impact of architectural inductive bias on training dynamics. The curve compares validation accuracy of \mname{} with and without Positional Encodings (PE), demonstrating that removing PE accelerates convergence and improves final performance.}
    \label{fig:pe_ablation}
\end{figure}
\captionsetup[figure]{skip=4pt}
Table~\ref{tab:ablation_on_architecture} evaluates encoder designs for NMR peak sets on \dataset.
Removing positional encodings from a vanilla Transformer improves Top-1 sequence accuracy by 6.08\% (27.28\% $\rightarrow$ 33.36\%), confirming that absolute position indices derived from arbitrary acquisition order, do not provide reliable structural cues for NMR spectra.

Replacing vanilla self-attention with \mname{}'s Multihead Attention Block (MAB/ISAB) yields a further gain of 7.69 percentage points (33.36\% $\rightarrow$ 41.05\%) without any global pooling mechanism.
\textbf{Note that while VT (w/o PE) is already permutation equivariant, the performance gain here highlights the specific advantage of the ISAB architecture.}
By mediating interactions through a small set of inducing points, ISAB creates an information bottleneck that effectively filters spurious pairwise correlations (attention dilution) inherent in standard full-attention mechanisms, thereby refining the spectral representation.

Adding PMA for global feature aggregation provides a marginal gain of 1.76 percentage points (41.05\% $\rightarrow$ 42.81\%), indicating that explicit global context extraction offers secondary refinement beyond the core set-equivariant representation.
These results demonstrate that matching architectural inductive bias to data structure—specifically, permutation equivariance and robust attention mechanisms—is critical for accurate structure elucidation.


\subsection{Generalization to Out-of-Distribution Datasets}
\label{sec:ood_generalization}

To assess robustness beyond our internal test set, we evaluated zero-shot generalization on two external benchmarks: NMRBank~\cite{Wang2025NMRExtractor} and a Multimodal Spectroscopic Dataset (MSD)~\cite{alberts2024unravelingmolecularstructuremultimodal}.

\begin{figure}[htbp]
    \centering
    \begin{subfigure}{0.42\linewidth}
        \centering
        \includegraphics[width=\linewidth]{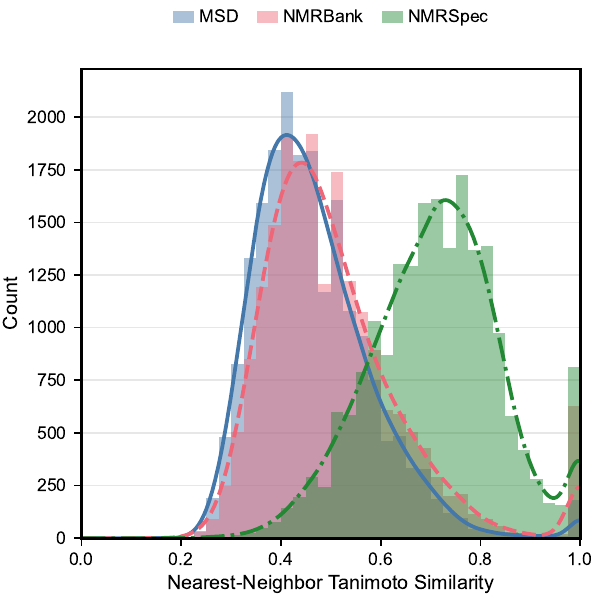}
        \caption{Domain Shift Analysis}
        \label{fig:ood_distribution}
    \end{subfigure}
    \hfill
    \begin{subfigure}{0.45\linewidth}
        \centering
        \includegraphics[width=\linewidth]{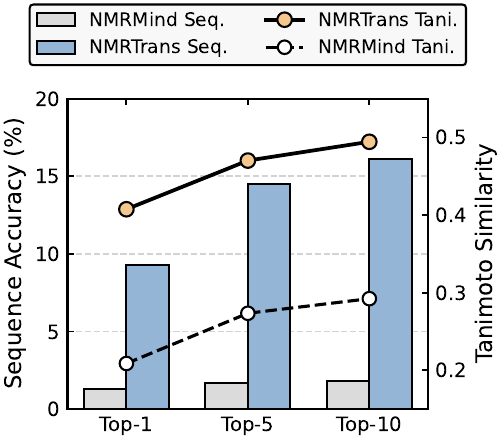}
        \caption{Performance on NMRBank}
        \label{fig:ood_nmrbank}
    \end{subfigure}
    
    \vspace{0.3em}
    
    \begin{subfigure}{\linewidth}
        \centering
        \includegraphics[width=0.9\linewidth]{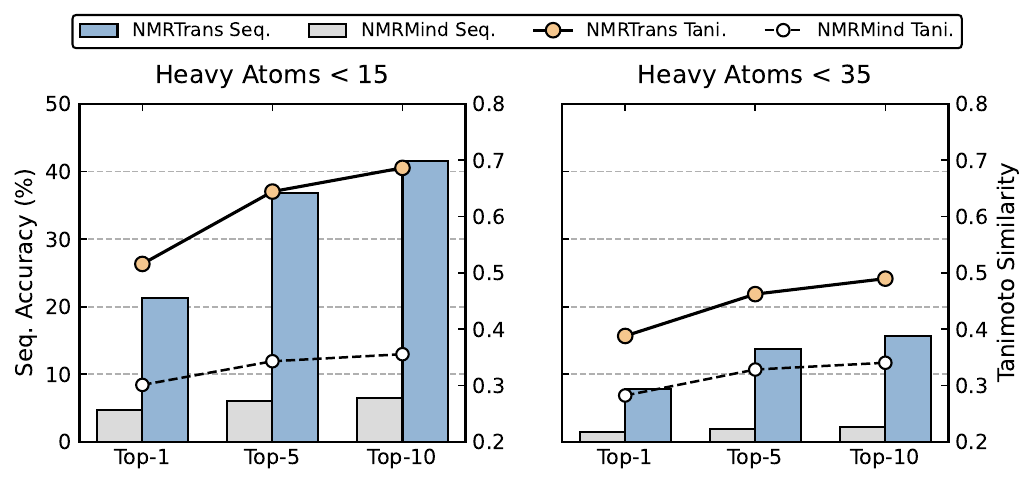}
        \caption{Performance on MSD~\cite{alberts2024unravelingmolecularstructuremultimodal}}
        \label{fig:ood_msd}
    \end{subfigure}
    
    \caption{Zero-shot generalization on Out-of-Distribution (OOD) datasets. 
    (a) Histogram of nearest-neighbor Tanimoto similarity between test samples and the training set. 
    (b) Performance on NMRBank. 
    (c) Performance on MSD, stratified by molecular size.}
    \label{fig:ood_results}
\end{figure}

\noindent\textbf{Domain shift quantification.}
Figure~\ref{fig:ood_distribution} visualizes the distribution of nearest-neighbor Tanimoto similarities between the OOD samples and our training set. 
The distribution skews significantly towards lower similarity values, confirming that these benchmarks contain structurally distinct molecules rather than subsets of the training distribution. This setup rigorously tests true zero-shot generalization capabilities.

\noindent\textbf{Performance on external benchmarks.} 
As shown in Figure~\ref{fig:ood_nmrbank}, \mname{} consistently outperforms the baseline on NMRBank despite the challenging domain shift, maintaining a clear lead in both sequence accuracy and structural similarity.
Results on the MSD benchmark (Figure~\ref{fig:ood_msd}) further validate this trend: \mname{} achieves higher accuracy across both small ($<15$ heavy atoms) and medium ($<35$ heavy atoms) subsets. 
These results collectively demonstrate that our permutation-invariant encoder generalizes effectively to diverse experimental conditions and molecular distributions unseen during training.
\section{Conclusion}
\label{sec:conclusion}

In this work, we addressed the challenge of molecular structure elucidation from NMR spectra under realistic experimental conditions, overcoming the long-standing scarcity of large-scale authentic training data.
To resolve this bottleneck, we introduced \dataset, a comprehensive corpus of experimental NMR spectra mined from chemical literature. This resource enables models to learn directly from real-world spectral distributions, thereby by passing the discrepancies inherent in computational simulations.
Building on \dataset, we proposed \mname, a set-based Transformer framework that explicitly aligns model inductive bias with the unordered and permutation-invariant nature of NMR peak sets.
Through extensive training and evaluation on experimental benchmarks, \mname\ demonstrates substantially improved accuracy and robustness compared to state-of-the-art generative and retrieval-based baselines.
Our results confirm that reliable structure elucidation in real-world scenarios requires the synergy of high-quality experimental supervision and physics-aware architectures. We believe \dataset\ and \mname\ together provide a solid foundation for scalable and practical NMR-based chemical analysis.



\bibliographystyle{ACM-Reference-Format}
\bibliography{main}

\clearpage

\appendix
\section*{Appendix Contents}
\startcontents[appendix]
\printcontents[appendix]{}{1}{}

\section{Related Work}

\subsection{Automated NMR Structure Elucidation}
NMR-based structure elucidation aims to infer molecular structures, typically represented as SMILES strings or molecular graphs, directly from NMR spectra.
Existing approaches to this task span a diverse set of modeling paradigms, differing not only in how molecular structures are represented and how the inverse mapping from spectra to structures is formulated, but also in how NMR spectral inputs are preprocessed and encoded.

Early data-driven methods introduced intermediate representations to reduce the complexity of the inverse problem. Representative substructure-intermediate frameworks first predict the presence of predefined functional groups or fragments from $^1$H and $^{13}$C spectra, and then assemble or rank candidate structures consistent with these predictions \cite{Huang2021}. 
Building on advances in sequence modeling, subsequent work reformulated structure elucidation as an end-to-end conditional generation task, encoding NMR spectra as sequences or sets and autoregressively decoding molecular SMILES using encoder--decoder Transformers \cite{Alberts2023, vita2024leveraging, nmr2struct}. 
To better capture structural uncertainty and molecular symmetries, recent studies have explored diffusion-based generative models that iteratively construct molecular graphs or atomic configurations conditioned on NMR observations \cite{yang2025diffnmr, xiong2025atomicdiffusion}. 
In parallel, search- and optimization-based approaches explicitly explore chemical space by combining spectral embeddings with discrete optimization techniques, such as genetic algorithms or large-scale nearest-neighbor retrieval followed by physics- or model-guided refinement \cite{Mirza2024, nmrsolver}. 
Despite their conceptual differences, a common limitation shared by most existing methods is their heavy reliance on simulated spectra for training, which often leads to degraded performance when applied to experimental NMR data.

\subsection{Simulated and Experimental NMR Datasets}
The effectiveness of learning-based structure elucidation methods is closely tied to the nature of the training data. Existing NMR datasets can be broadly categorized into DFT-simulated, machine-learning-simulated, and experimental collections. 
DFT-based datasets, such as QM9-NMR \cite{Gupta2021}, provide high-fidelity spectra at moderate scale and are widely used for training forward NMR prediction models, while machine-learning-simulated datasets further extend coverage to tens or hundreds of millions of molecules by predicting chemical shifts at scale. Although these simulated resources enable large-scale training and retrieval, they inevitably inherit modeling biases and idealized acquisition assumptions, which can limit their transferability to real experimental settings.

In contrast, experimental NMR datasets capture authentic measurement variability arising from solvent effects, noise, coupling complexity, and instrument-dependent artifacts. Early community-driven resources such as NMRShiftDB~\cite{Steinbeck2003} and more recent curated collections like NMRBank~\cite{Wang2025} provide experimentally grounded molecule--spectrum pairs, while larger-scale efforts have begun to aggregate multimodal experimental data for systematic benchmarking~\cite{alberts2024unravelingmolecularstructuremultimodal}, such as IR. UV, and Raman spectroscopy. 
Despite these advances, large-scale experimental benchmarks remain scarce. 

Recent work has begun to alleviate the lack of large-scale experimental NMR data.
NMRExp~\cite{Wang2025} introduces, to date, the largest collection of experimentally acquired NMR spectra, consisting of over 3.3 million records mined from the supporting information of chemical literature.
Similar to our approach, NMRExp relies on large-scale literature mining rather than simulated spectra, but primarily focuses on dataset construction, with a scale substantially larger than ours.
Concurrently, NMRGym~\cite{fang2026nmrgymcomprehensivebenchmarknuclear} proposes a standardized benchmark built upon high-quality experimental $^1$H and $^{13}$C NMR spectra, integrating curated datasets with unified evaluation protocols to facilitate fair comparison across structure elucidation methods.
Together, these efforts underscore both the growing availability of experimental NMR data and the continued need for methods that can effectively learn directly from experimental spectra.

\subsection{Language Modeling for NMR Elucidation}
\label{subsec:modeling_paradigms}

Existing approaches for NMR structure elucidation typically reformulate the inverse problem as conditional generation, falling into two primary categories:

\noindent\textbf{Autoregressive Models} factorize the probability of a SMILES string $\mathbf{x}$ given a spectrum $\mathbf{y}$ as $p(\mathbf{x}|\mathbf{y}) = \prod_{t} p(x_t | x_{<t}, \mathbf{y})$. These methods often serialize spectral features into a pseudo-sequence to apply standard Encoder-Decoder Transformers, implicitly imposing an arbitrary order on the input peaks. Representative approaches include NMR2Struct~\cite{nmr2struct} and NMRMind~\cite{xue2025nmrmind}, which employ Transformer-based architectures for conditional SMILES generation, as well as T5-based methods~\cite{vita2024leveraging} that serialize 2D NMR peaks into sequences for encoder-decoder generation.

\noindent\textbf{Diffusion Models} generate structures via iterative denoising, learning a distribution $p_\theta(\mathbf{x}_0 | \mathbf{y})$ by reversing a noise process. DiffNMR~\cite{yang2025diffnmr} implements discrete diffusion for generating molecular 2D graphs $G = (V, E)$, perturbing node and edge types separately. In contrast, ChefNMR~\cite{xiong2025atomicdiffusion} employs continuous diffusion in 3D coordinate space, directly generating atomic conformations from noisy coordinates $X_\sigma = X_0 + n$ conditioned on NMR embeddings, capturing geometric symmetries of molecular structures. DiffSpectra~\cite{wang2025diffspectra} explores the potential of diffusion models for molecular structure elucidation from vibrational spectra, including IR, UV-Vis, and Raman spectroscopy. Although these modalities differ fundamentally from NMR, their diffusion-based Transformer architecture demonstrates consistent performance gains, suggesting that such generative frameworks may also be effective for permutation-invariant spectral data like NMR when properly adapted~\cite{yang2025diffnmr}.

\section{Theoretical Proof of End-to-End Permutation Invariance}
\label{app:permutation}

\begin{figure*}[t]
    \centering
    \includegraphics[width=0.9\linewidth]{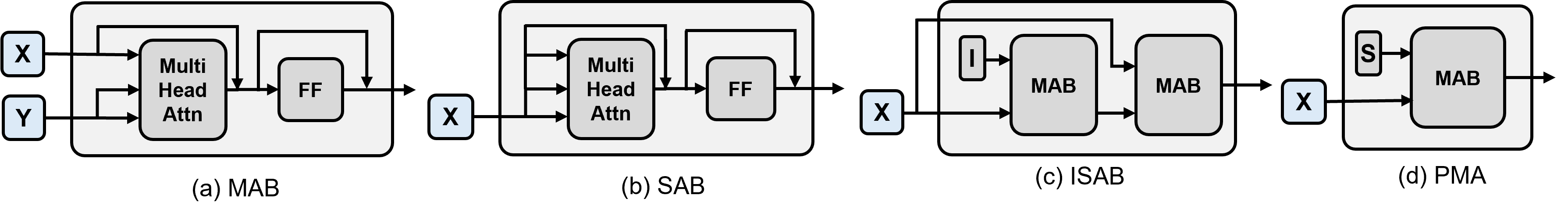}
    \caption{Core building blocks of Set Transformer~\cite{lee2019set}. 
    (a) Multihead Attention Block (MAB) performs attention between two sets.
    (b) Self-Attention Block (SAB) is a special case of MAB with shared inputs.
    (c) Induced Set Attention Block (ISAB) introduces learnable inducing points to reduce computational complexity. 
    (d) Pooling by Multihead Attention (PMA) performs set pooling using learnable seed vectors. }
    \Description{Diagram showing four architectural blocks of Set Transformer arranged horizontally: (a) MAB with two input sets $X$ and $Y$ feeding into multi-head attention producing output; (b) SAB showing a single set $X$ attending to itself; (c) ISAB with input set $X$ attending to a smaller set of inducing points $I$, then back to $X$; (d) PMA with a small set of seed vectors S attending to input set $X$.}
    \label{fig:set_transformer_blocks}
\end{figure*}

This appendix provides formal proofs that our architecture eliminates spurious position-chemistry correlations by design. We establish that the spectral encoder is \textit{permutation equivariant} and the subsequent decoder conditioning is \textit{permutation invariant}.

\subsection{Preliminaries and Definitions}

Let the input NMR spectrum be a set of $n$ peaks represented as a feature matrix $\mathbf{X} \in \mathbb{R}^{n \times d}$. A permutation $\pi$ of indices $\{1, \dots, n\}$ is represented by a permutation matrix $\mathbf{\Pi}_\pi \in \{0,1\}^{n \times n}$ where:
\[ (\mathbf{\Pi}_\pi)_{ij} = \delta_{i,\pi(j)} \]
Applying $\pi$ to the input yields $\mathbf{X}' = \mathbf{\Pi}_\pi \mathbf{X}$. We define the following properties for a function $f$:
\begin{itemize}
    \item \textbf{Permutation Invariance:} $f(\mathbf{\Pi}_\pi \mathbf{X}) = f(\mathbf{X})$
    \item \textbf{Permutation Equivariance:} $f(\mathbf{\Pi}_\pi \mathbf{X}) = \mathbf{\Pi}_\pi f(\mathbf{X})$
\end{itemize}

\subsection{Properties of Attention Primitives}

The core building block is the Multihead Attention Block, $\mathrm{MAB}(\mathbf{Q}, \mathbf{K}, \mathbf{V})$. For brevity, we focus on the core mechanism:
\[ \mathrm{Attn}(\mathbf{Q}, \mathbf{K}, \mathbf{V}) = \mathrm{softmax}\left(\frac{\mathbf{Q}\mathbf{K}^\top}{\sqrt{d}}\right)\mathbf{V} \]
where $\mathrm{softmax}$ operates independently on each row of its input matrix.

\begin{lemma}[Permutation Properties of MAB]
Let $\mathbf{Q} \in \mathbb{R}^{n \times d}$, $\mathbf{K} \in \mathbb{R}^{m \times d}$, and $\mathbf{V} \in \mathbb{R}^{m \times d}$. The MAB satisfies:
\begin{enumerate}
    \item \textbf{Equivariance w.r.t. Query:} $\mathrm{MAB}(\mathbf{\Pi}_\pi \mathbf{Q}, \mathbf{K}, \mathbf{V}) = \mathbf{\Pi}_\pi \mathrm{MAB}(\mathbf{Q}, \mathbf{K}, \mathbf{V})$
    \item \textbf{Invariance under Synchronized Permutation of Keys and Values:} 
    $\mathrm{MAB}(\mathbf{Q}, \mathbf{\Pi}_\pi \mathbf{K}, \mathbf{\Pi}_\pi \mathbf{V}) = \mathrm{MAB}(\mathbf{Q}, \mathbf{K}, \mathbf{V})$
\end{enumerate}
\end{lemma}

\begin{proof}
\textbf{For Property (1):} Let $\mathbf{Q}' = \mathbf{\Pi}_\pi \mathbf{Q}$. Since $\mathbf{\Pi}_\pi$ permutes the rows of $\mathbf{Q}$ and softmax operates independently on each row, we have:
\[ \mathrm{softmax}\left(\frac{(\mathbf{\Pi}_\pi \mathbf{Q})\mathbf{K}^\top}{\sqrt{d}}\right) = \mathbf{\Pi}_\pi \, \mathrm{softmax}\left(\frac{\mathbf{Q}\mathbf{K}^\top}{\sqrt{d}}\right) \]
Thus the output transforms as:
\[ \mathrm{Attn}(\mathbf{\Pi}_\pi \mathbf{Q}, \mathbf{K}, \mathbf{V}) = \mathbf{\Pi}_\pi \, \mathrm{softmax}\left(\frac{\mathbf{Q}\mathbf{K}^\top}{\sqrt{d}}\right)\mathbf{V} = \mathbf{\Pi}_\pi \, \mathrm{Attn}(\mathbf{Q}, \mathbf{K}, \mathbf{V}) \]

\textbf{For Property (2):} Let $\mathbf{K}' = \mathbf{\Pi}_\pi \mathbf{K}$ and $\mathbf{V}' = \mathbf{\Pi}_\pi \mathbf{V}$. The attention scores become:
\[ \mathbf{A}' = \mathrm{softmax}\left(\frac{\mathbf{Q}(\mathbf{\Pi}_\pi \mathbf{K})^\top}{\sqrt{d}}\right) = \mathrm{softmax}\left(\frac{\mathbf{Q}\mathbf{K}^\top \mathbf{\Pi}_\pi^\top}{\sqrt{d}}\right) \]
Since $\mathbf{\Pi}_\pi^\top$ permutes the \emph{columns} of $\mathbf{Q}\mathbf{K}^\top$, the softmax output satisfies $\mathbf{A}' = \mathbf{A} \mathbf{\Pi}_\pi^\top$. The output is then:
\[ \begin{aligned}
\mathrm{Attn}(\mathbf{Q}, \mathbf{\Pi}_\pi \mathbf{K}, \mathbf{\Pi}_\pi \mathbf{V}) 
&= \mathbf{A}' (\mathbf{\Pi}_\pi \mathbf{V}) \\
&= (\mathbf{A} \mathbf{\Pi}_\pi^\top) (\mathbf{\Pi}_\pi \mathbf{V}) \\
&= \mathbf{A} (\mathbf{\Pi}_\pi^\top \mathbf{\Pi}_\pi) \mathbf{V} = \mathbf{A} \mathbf{V} \\
&= \mathrm{Attn}(\mathbf{Q}, \mathbf{K}, \mathbf{V})
\end{aligned} \]
where we used $\mathbf{\Pi}_\pi^\top \mathbf{\Pi}_\pi = \mathbf{I}$. Hence the aggregation is invariant under synchronized permutation of keys and values.
\end{proof}

\subsection{Proof of Encoder Properties}

\begin{proposition}[ISAB is Equivariant]
$\mathrm{ISAB}(\mathbf{\Pi}_\pi \mathbf{X}) = \mathbf{\Pi}_\pi \mathrm{ISAB}(\mathbf{X})$.
\end{proposition}

\begin{proof}
ISAB uses inducing points $\mathbf{I}$ and input $\mathbf{X}$ through two steps:
\begin{enumerate}
    \item \textbf{Inducing Step:} $\mathbf{H} = \mathrm{MAB}(\mathbf{I}, \mathbf{X}, \mathbf{X})$. Since $\mathbf{X}$ serves as both Key and Value, by Lemma 1.2 with synchronized permutation:
    \[ \mathbf{H}(\mathbf{\Pi}_\pi \mathbf{X}) = \mathrm{MAB}(\mathbf{I}, \mathbf{\Pi}_\pi \mathbf{X}, \mathbf{\Pi}_\pi \mathbf{X}) = \mathbf{H}(\mathbf{X}) \]
    \item \textbf{Refining Step:} $\mathbf{Z} = \mathrm{MAB}(\mathbf{X}, \mathbf{H}, \mathbf{H})$. Since $\mathbf{X}$ is the Query, by Lemma 1.1:
    \[ \mathrm{MAB}(\mathbf{\Pi}_\pi \mathbf{X}, \mathbf{H}, \mathbf{H}) = \mathbf{\Pi}_\pi \mathrm{MAB}(\mathbf{X}, \mathbf{H}, \mathbf{H}) = \mathbf{\Pi}_\pi \mathbf{Z} \]
\end{enumerate}
Thus, the composition is equivariant.
\end{proof}

\begin{proposition}[PMA is Invariant]
$\mathrm{PMA}(\mathbf{\Pi}_\pi \mathbf{X}) = \mathrm{PMA}(\mathbf{X})$.
\end{proposition}

\begin{proof}
PMA is defined as $\mathrm{MAB}(\mathbf{S}, \mathbf{X}, \mathbf{X})$ where $\mathbf{S}$ is a fixed seed matrix independent of input ordering. Applying Lemma 1.2 with synchronized permutation of keys and values yields the invariance directly.
\end{proof}

\subsection{Proof of Decoder Interface Invariance}

\begin{lemma}[Cross-Attention Invariance]
Let $\mathbf{Z}$ be equivariant features ($\mathbf{Z}_\pi = \mathbf{\Pi}_\pi \mathbf{Z}$), and let the decoder query $\mathbf{q}_t$ depend only on the previously generated tokens $\mathbf{y}_{<t}$ and permutation-invariant features (e.g., $\mathbf{g}$). Then:
\[ \mathrm{CrossAttn}(\mathbf{q}_t, \mathbf{Z}_\pi, \mathbf{Z}_\pi) = \mathrm{CrossAttn}(\mathbf{q}_t, \mathbf{Z}, \mathbf{Z}) \]
\end{lemma}

\begin{proof}
The context vector $\mathbf{c}_t$ is computed as:
\[ \mathbf{c}_t = \sum_{j=1}^n \alpha_{tj} \mathbf{z}_j, \quad \text{where } \alpha_{tj} = \frac{\exp(\mathbf{q}_t^\top \mathbf{z}_j / \sqrt{d})}{\sum_k \exp(\mathbf{q}_t^\top \mathbf{z}_k / \sqrt{d})} \]
Under permutation $\pi$, the features become $\mathbf{z}'_j = \mathbf{z}_{\pi(j)}$. The new attention weights are:
\[ \alpha'_{tj} = \frac{\exp(\mathbf{q}_t^\top \mathbf{z}_{\pi(j)} / \sqrt{d})}{\sum_k \exp(\mathbf{q}_t^\top \mathbf{z}_{\pi(k)} / \sqrt{d})} = \alpha_{t,\pi(j)} \]
The resulting context vector is:
\[ \mathbf{c}'_t = \sum_{j=1}^n \alpha'_{tj} \mathbf{z}'_j = \sum_{j=1}^n \alpha_{t,\pi(j)} \mathbf{z}_{\pi(j)} = \sum_{i=1}^n \alpha_{ti} \mathbf{z}_i = \mathbf{c}_t \]
where the last equality follows from reindexing $i = \pi(j)$ and the commutativity of addition. Hence the context vector is invariant.
\end{proof}

\subsection{Main Theorem}

\begin{theorem}[End-to-End Order Independence]
The architecture satisfies $p(\mathbf{y} \mid \mathbf{\Pi}_\pi \mathbf{X}) = p(\mathbf{y} \mid \mathbf{X})$ for any permutation $\pi$.
\end{theorem}

\begin{proof}
The generation process depends on two encoder outputs:
\begin{itemize}
    \item $\mathbf{Z}$: output of ISAB $\implies$ permutation equivariant (Proposition 1)
    \item $\mathbf{g}$: output of PMA $\implies$ permutation invariant (Proposition 2)
\end{itemize}
At each decoding step $t$, the query $\mathbf{q}_t$ is a deterministic function of the previously generated tokens $\mathbf{y}_{<t}$ and the invariant global feature $\mathbf{g}$. By Lemma 2, the cross-attention context $\mathbf{c}_t = \mathrm{CrossAttn}(\mathbf{q}_t, \mathbf{Z}, \mathbf{Z})$ is invariant to permutations of $\mathbf{Z}$. Consequently, the input to the decoder's prediction layer—comprising $\mathbf{q}_t$, $\mathbf{c}_t$, and $\mathbf{g}$—is independent of the initial peak ordering. Since the decoder itself is a deterministic function, the conditional distribution $p(\mathbf{y} \mid \mathbf{X})$ satisfies permutation invariance.
\end{proof}

\section{Dataset Curation and Analysis}
\label{app:dataset_details}

This appendix details the construction pipeline and statistical characteristics of \textbf{\dataset}. The curation process was designed to systematically harvest high-quality experimental spectroscopy data from unstructured Supporting Information (SI) documents.

\subsection{Detailed Curation Pipeline}
The extraction workflow consists of five sequential stages, combining layout analysis, large language models (LLMs), and rule-based parsing to ensure high precision and recall.

\noindent\textbf{Stage 1: Document Preprocessing.} 
Raw PDF documents were first converted to high-resolution images (300 DPI) to preserve visual fidelity. Simultaneously, we employed \texttt{MinerU} to parse the document layout and extract textual content into structured Markdown format. To handle long documents efficiently, text sections exceeding 8,192 characters were processed using a sliding window approach with a 3-line overlap, ensuring contextual continuity across segments.

\noindent\textbf{Stage 2: Spectroscopic Text Extraction.} 
We utilized LLMs to identify text blocks containing spectroscopic data. A specialized rule-based parser using strict regular expressions was then applied to extract NMR parameters. The parser is designed to handle diverse reporting formats, including:
\begin{itemize}
    \item Range patterns with tilde or space separators (e.g., ``7.30$\sim$7.35'').
    \item Standard single-value patterns with multiplicities (e.g., ``7.32 (d, $J=7.2$ Hz, 1H)'').
    \item Bare chemical shift values.
\end{itemize}
Complex cases such as major/minor isomer mixtures were identified and processed separately. Extensive text normalization was performed to standardize chemical symbols and character encodings.

\noindent\textbf{Stage 3: Data Filtering \& Validation.} 
Extracted entries underwent rigorous validation. We checked for metadata completeness (e.g., presence of compound identifiers) and validated IUPAC names via iterative LLM verification. Duplicate entries sharing identical compound identifiers were deduplicated, prioritizing records with the most complete spectroscopic information.

\noindent\textbf{Stage 4: SMILES Conversion.} 
Valid IUPAC names were converted to SMILES strings using OPSIN in batches with timeout protection. The resulting SMILES were standardized and canonicalized using RDKit to ensure chemical validity and uniqueness.

\noindent\textbf{Stage 5: Spectroscopic Image Recognition.} 
For entries with associated spectrum images, we employed vision-language models to extract metadata (e.g., frequency, solvent). To prevent information leakage during model training, molecular structures appearing within the spectrum images were automatically detected and masked (whitened) using a YOLO-based object detection model.

\subsection{Dataset Statistics}

In this section, we provide detailed statistics on the \textbf{\dataset} corpus (2013--2025 subset). Table~\ref{tab:temporal_dist} presents the temporal distribution of the collected data, demonstrating consistent coverage of contemporary literature. Table~\ref{tab:spectrum_dist} details the distribution of spectroscopic modalities in the source collection, highlighting the diversity of the dataset.

\begin{table}[!hbtp]
    \small
    \centering
    \caption{Temporal Distribution of Molecular Data (2013--2025).}
    \label{tab:temporal_dist}
    \begin{tabular}{lrrrr}
        \toprule
        \textbf{Year} & \textbf{Docs} & \textbf{Molecules} & \textbf{Spectra} & \textbf{Avg.} \\
        \midrule
        2013 & 5,806 & 56,118 & 186,278 & 3.32 \\
        2014 & 4,496 & 42,439 & 138,528 & 3.26 \\
        2015 & 5,216 & 51,669 & 167,299 & 3.24 \\
        2016 & 5,915 & 60,942 & 197,464 & 3.24 \\
        2017 & 7,151 & 79,317 & 257,126 & 3.24 \\
        2018 & 9,239 & 102,910 & 330,830 & 3.21 \\
        2019 & 9,439 & 103,486 & 321,149 & 3.10 \\
        2020 & 7,932 & 91,141 & 282,041 & 3.09 \\
        2021--25 & 6,877 & 93,968 & 263,777 & 2.81 \\
        \midrule
        \textbf{Total} & \textbf{62,071} & \textbf{681,990} & \textbf{2,144,492} & \textbf{3.15} \\
        \bottomrule
    \end{tabular}
\end{table}

\begin{table}[hbtp]
    \centering
    \caption{Distribution of Spectroscopic Data Types.}
    \label{tab:spectrum_dist}
    \begin{tabular}{lrr}
        \toprule
        \textbf{Type} & \textbf{Count} & \textbf{Pct (\%)} \\
        \midrule
        $^1$H NMR & 1,235,881 & 29.2 \\
        $^{13}$C NMR & 1,038,831 & 24.5 \\
        Mass Spec & 1,066,046 & 25.2 \\
        IR & 495,083 & 11.7 \\
        $^{19}$F NMR & 110,273 & 2.6 \\
        HPLC & 90,015 & 2.1 \\
        $^{31}$P NMR & 36,061 & 0.9 \\
        $^{11}$B NMR & 18,766 & 0.4 \\
        Other & 148,249 & 3.5 \\
        \bottomrule
    \end{tabular}
\end{table}

\subsection{Comparison with NMRExp~\cite{Wang2025}}
\label{subsec:dataset_comparison}

To assess the distinctiveness and value of our curated corpus, we performed a comprehensive comparative analysis against \textbf{NMRExp}, a \textbf{concurrent work} that also aggregates large-scale NMR data. This comparison highlights the unique contribution of our pipeline regarding data quality, molecular exclusivity, and chemical space coverage.

\noindent\textbf{Data Uniqueness and Quality.}
A critical distinction lies in the curation standards. While NMRExp aggregates a massive volume of raw entries ($3.37$ million), it exhibits significant redundancy with an internal uniqueness ratio of only 46.5\% ($1.57$ million unique SMILES). In contrast, our pipeline emphasizes quality over raw quantity. The final \dataset{} subset consists of \textbf{239,997} samples which are \textbf{100\% unique} molecules, having undergone rigorous deduplication and validity checks.
 A random sample of 100 unique compounds was selected, ensuring diversity in molecular complexity and spectral types. This yielded 200 spectra (100 $^1H$ and 100 $^{13}C$). Evaluation Metrics:
Compound-level accuracy: The percentage of compounds for which all associated spectral data (including SMILES, solvent, frequency, and peak lists) were correctly extracted.
Spectrum-level accuracy: The percentage of individual spectra with fully accurate peak lists and metadata.
Image integrity: The accuracy of spectral image-to-compound ID linkage and visual completeness. It reached 96.5\%, 93.0\%, and 92\% respectively.
The high accuracy rates, particularly at the spectrum and compound levels, validate the robustness of our data extraction methodology. The primary source of errors was traced to ambiguities in legacy PDF formatting, which future iterations of the pipeline will address with enhanced layout analysis.


In addition, \dataset{} uniquely provides aligned spectroscopic \textbf{\emph{images}} alongside structured textual data, whereas NMRExp contains no spectrum images, enabling joint vision–language modeling and multimodal evaluation, even though the image modality is not exploited in the present study.

\noindent\textbf{Overlap and Complementarity.}
We analyzed the intersection of molecular structures between the two datasets to determine if \dataset{} merely subsets this concurrent repository. As illustrated in Figure~\ref{fig:venn_overlap}, the analysis reveals a remarkably low overlap:
\begin{itemize}
    \item Only \textbf{58,796} molecules are shared between the two datasets, representing a Jaccard overlap index of merely \textbf{3.36\%}.
    \item Crucially, \textbf{181,201} molecules (75.5\% of our dataset) are \textbf{exclusive} to \dataset{}.
\end{itemize}
This high proportion of unique entries demonstrates that mining Supplementary Information (SI) captures a distinct segment of chemical space often overlooked by other collection methods, providing significant complementary value to the community.

\noindent\textbf{Chemical Space Properties.}
Table~\ref{tab:comparison_stats} compares the physicochemical properties of the molecules in both datasets. The average molecular weight (MW) and atom counts are comparable, with \dataset{} averaging 327.96 Da and NMRExp averaging 340.58 Da. However, NMRExp includes a long tail of extremely large molecules (Max MW $>2200$, Max Atoms 172), whereas \dataset{} shows a more concentrated distribution (Max MW 944, Max Atoms 52). This suggests that our dataset is particularly focused on medium-sized organic molecules typical of medicinal chemistry and total synthesis research, ensuring stability during model training by filtering out overly complex outliers.

\begin{figure}
\begin{tikzpicture} 
    
    \definecolor{myteal}{RGB}{0, 128, 128}
    \definecolor{mygray}{RGB}{105, 105, 105}
    \definecolor{myorange}{RGB}{204, 102, 0}

    \filldraw[fill=myteal!15, draw=myteal, thick] (0,0) circle (1.8cm);
    
    \filldraw[fill=mygray!15, draw=mygray, thick] (3.0,0) circle (2.1cm);

    \begin{scope}
        \clip (0,0) circle (1.8cm);
        \fill[myorange!30] (3.0,0) circle (2.1cm);
    \end{scope}
    
    
    \node[text=myteal, font=\bfseries\large] at (-0.2, 1.3) {\dataset}; 
    
    \node[align=center, font=\small] at (-0.6, 0) {
        Unique to \dataset \\ 
        \textbf{181,201}
    };

    \node[text=mygray, font=\bfseries\large] at (3.3, 1.3) {NMRExp};
    
    \node[align=center, font=\small] at (3.3, 0) {
        Unique to NMRExp \\ 
        $\sim$1.51M
    };

    \node[align=center, font=\footnotesize, text=myorange!100!black] at (1.3, 0) {
        Overlap \\ 
        \textbf{58,796}
    };

\end{tikzpicture}
\caption{NMR spectroscopy datasets overlap between \dataset\ and NMRExp}
\label{fig:venn_overlap}
\end{figure}
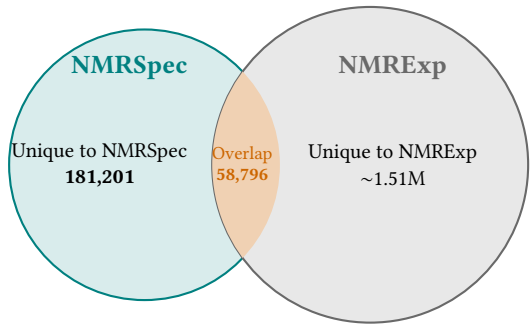

\begin{table}[hbtp]
    \centering
    \caption{Statistical comparison between \dataset{} and the concurrent work NMRExp. \dataset{} exhibits 100\% internal uniqueness and focuses on a tightly distributed range of medium-sized organic molecules.}
    \label{tab:comparison_stats}
    \small
    \begin{tabular}{lrr}
        \toprule
        \textbf{Metric} & \textbf{\dataset (Ours)} & \textbf{NMRExp} \\
        \midrule
        \multicolumn{3}{l}{\textit{Dataset Size \& Uniqueness}} \\
        Total Records & 239,997 & 3,372,987 \\
        Unique Molecules & \textbf{239,997} & 1,568,289 \\
        Internal Uniqueness & \textbf{100.0\%} & 46.5\% \\
        \midrule
        \multicolumn{3}{l}{\textit{Overlap Analysis}} \\
        Shared Molecules & \multicolumn{2}{c}{58,796} \\
        Exclusive Molecules & \textbf{181,201 (75.5\%)} & 1,509,493 (96.2\%) \\
        \midrule
        \multicolumn{3}{l}{\textit{Molecular Properties (Avg $\pm$ Std)}} \\
        Molecular Weight (Da) & $327.96 \pm 104.3$ & $340.58 \pm 124.1$ \\
        Heavy Atom Count & $22.99 \pm 7.36$ & $23.75 \pm 8.80$ \\
        Ring Count & $2.70$ (Median: 3) & $2.84$ (Median: 3) \\
        Max Heavy Atoms & 52 & 172 \\
        \bottomrule
    \end{tabular}
\end{table}

\section{Additional Experimental Analysis}
\label{app:additional_experiments}

\subsection{Impact of Model Scale}

In this section, we investigate the scaling behavior of \mname{} by evaluating different backbone sizes of the T5 architecture: \texttt{T5-small} (60M parameters), \texttt{T5-base} (220M parameters), and \texttt{T5-large} (770M parameters). 

\noindent\textbf{Scaling Analysis.} 
Table~\ref{tab:model_scaling_results} presents the quantitative performance across these variations. Interestingly, we observe that increasing the model size from \texttt{small} to \texttt{large} does not yield a significant improvement in sequence accuracy or Tanimoto similarity. For instance, the Top-1 sequence accuracy for \texttt{T5-large} is only marginally better (or comparable) to that of \texttt{T5-base}. 

This lack of substantial scaling gain suggests that for the current complexity of experimental NMR structure elucidation, the performance is primarily bottlenecked by the \emph{availability and diversity of experimental data} rather than the representation capacity of the transformer backbone. While \dataset is an order of magnitude larger than previous open-access databases, it remains relatively sparse compared to the billions of tokens typically used to saturate large-scale T5 models in natural language processing. These findings underscore that matching architectural inductive biases (e.g., set-based encoding) and ensuring data high-fidelity are more critical than pure parameter scaling in this domain.

\noindent\textbf{Implement Details.} 
All model variants (\texttt{T5-small}, \texttt{T5-base}, and \texttt{T5-large}) employ the encoder architecture described in Section~\ref{sec:methodology}. Training uses a fixed learning rate of $1\times10^{-4}$, gradient accumulation steps of 1, gradient clipping with maximum norm 1.0, and 32-bit floating-point precision. To maintain a consistent effective batch size of 2048 across scales, we dynamically adjust the per-GPU batch size according to available memory: \texttt{T5-small} uses 2 GPUs with batch size 1024 per GPU, \texttt{T5-base} uses 4 GPUs with batch size 512 per GPU, and \texttt{T5-large} uses 8 GPUs with batch size 256 per GPU. All experiments are conducted on NVIDIA H200 (141 GB) nodes. Total computational cost ranges from approximately 100 GPU hours for \texttt{T5-small} to 300 GPU hours for \texttt{T5-large}.

\begin{table}[h]
    \small
    \centering
    \caption{Performance comparison of \mname{} across different model scales. All variants use the same set-transformer encoder and are evaluated on the \dataset{} test set.}
    \label{tab:model_scaling_results}
    \begin{tabular}{lcccc}
        \toprule
        \textbf{Backbone} & \textbf{Parameters} & \textbf{Top-1 Acc} & \textbf{Top-1 Tani} & \textbf{Top-1 Tok} \\
        \midrule
        T5-small & 60M  & 42.81 & 74.62 & 71.30 \\
        T5-base  & 220M & 42.71 & 74.52 & 71.37 \\
        T5-large & 770M & 30.80 & 63.91 & 62.31 \\
        \bottomrule
    \end{tabular}
\end{table}

\noindent\textbf{Training Dynamics.} 
Figure~\ref{fig:scaling_curves} illustrates the training and validation loss curves for the three model sizes. All models exhibit similar convergence patterns, further confirming that the additional parameters in larger variants do not lead to superior structural understanding under the current data regime. The validation accuracy plateaus at comparable levels across all scales, indicating that structural ambiguity in experimental spectra—rather than model capacity—is the limiting factor.

\begin{figure}[h]
    \centering
    \begin{minipage}[b]{0.48\linewidth}
        \centering
        \includegraphics[width=\linewidth]{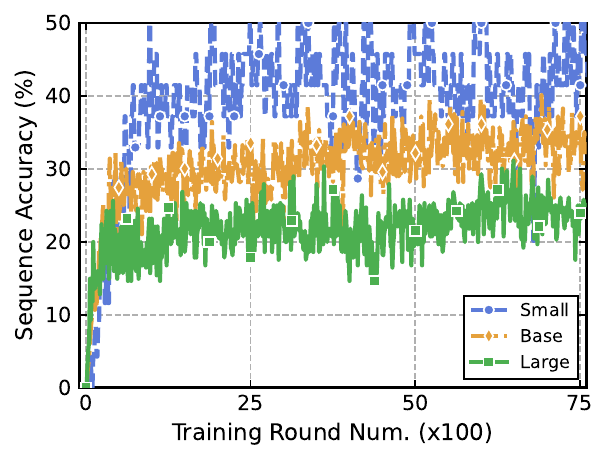}
        \small (a) Sequence accuracy
    \end{minipage}
    \hfill
    \begin{minipage}[b]{0.48\linewidth}
        \centering
        \includegraphics[width=\linewidth]{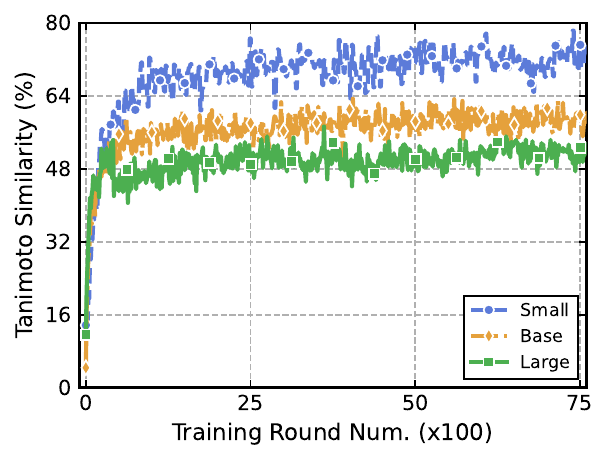}
        \small (b) Tanimoto similarity
    \end{minipage}
    \caption{Training dynamics across model scales. Contrary to conventional scaling laws, the small model consistently achieves the highest performance in both sequence accuracy and molecular similarity metrics, followed by the base model, while the large model exhibits the weakest generalization capability—suggesting potential optimization challenges or overfitting in larger architectures for this spectral-to-structure task.}
    \label{fig:scaling_curves}
\end{figure}

\subsection{Encoder Equivariance Analysis}
\label{app:model_scaling}

\begin{figure}[!h]
\small
\centering
\includegraphics[width=0.95\linewidth]{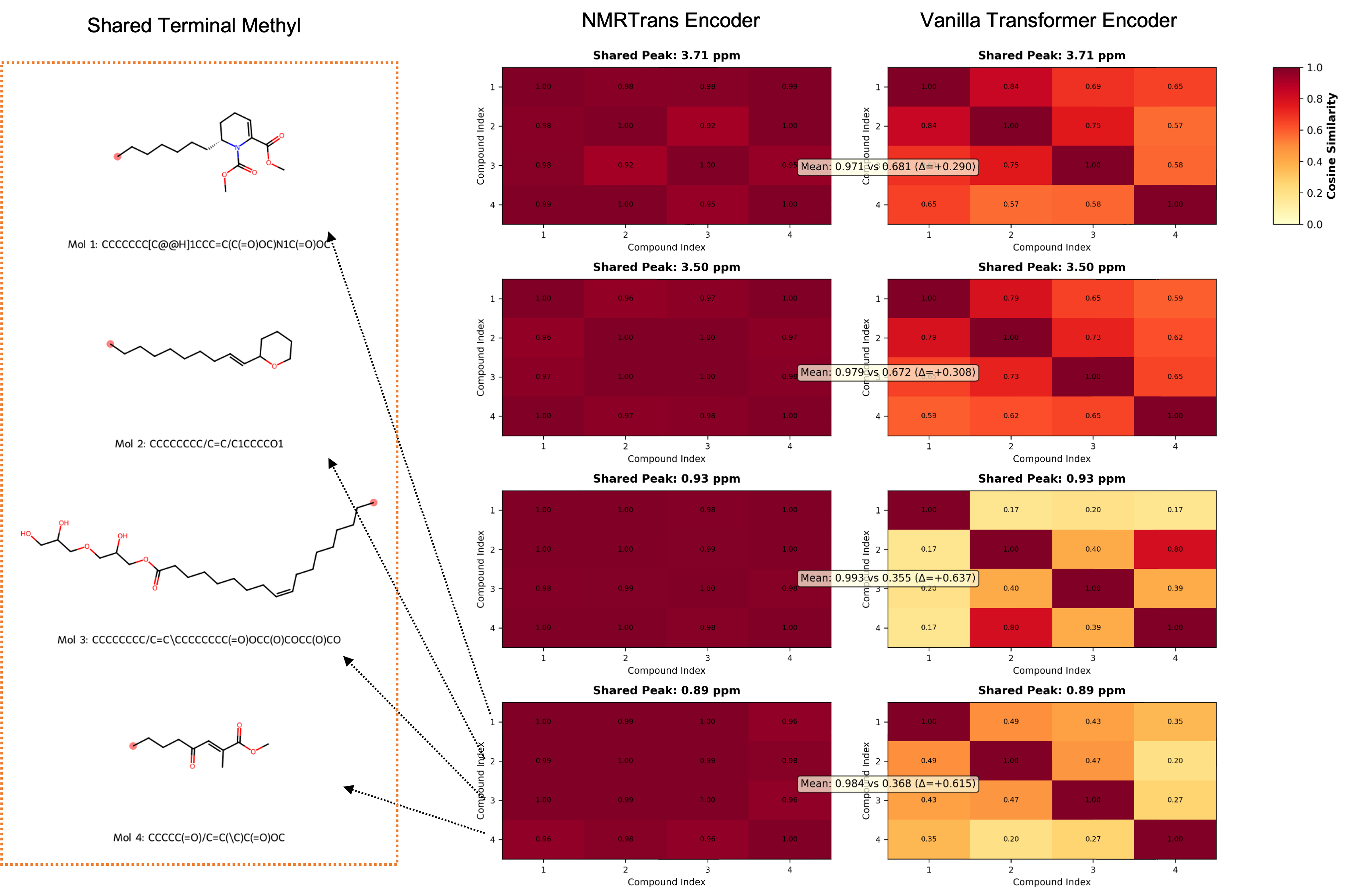}
\caption{Permutation equivariance validation via feature similarity matrices of shared $^1$H NMR peaks. We analyze four representative compounds (Mol 1--4) containing conserved \textbf{ester, alkyl, and cyclic ether motifs}. For each shared peak, \textbf{NMRTrans Encoder (left column)} produces near-identity similarity matrices ($\mu \geq 0.971$), confirming strict permutation equivariance where chemically equivalent peaks yield identical representations regardless of input order. In contrast, the \textbf{Vanilla Transformer Encoder (right column)} exhibits severe position-dependent variations due to positional encodings. This disparity is evident for the methoxy peak at $\delta$~3.71~ppm, where similarity drops from 0.971 to 0.681 ($\Delta = 0.290$), and is even more pronounced for terminal methyl protons at $\delta$~0.89~ppm, where the baseline's similarity collapses to 0.368.}
\label{fig:equivariance}
\end{figure}

To validate the permutation equivariance of our NMRTrans encoder (based on the Set Transformer architecture), we analyze shared $^1$H NMR peaks across compound groups with conserved structural motifs. For a set of molecules sharing specific chemical environments—such as the methoxy group at $\delta \approx 3.71$~ppm or terminal methyls at $\delta \approx 0.89$~ppm—we extract the encoder output features $z_h$ corresponding to these peaks and compute pairwise cosine similarities across the compounds.

As shown in Figure~\ref{fig:equivariance}, NMRTrans consistently produces near-uniform similarity matrices where off-diagonal elements remain remarkably high (mean $\mu \geq 0.971$ across all groups), confirming that feature representations depend solely on chemical attributes rather than input ordering. In stark contrast, the Vanilla Transformer exhibits significant position-dependent variations. Its reliance on positional encodings artificially differentiates identical chemical environments based on their sequence position, violating the fundamental requirement for physically meaningful spectral modeling.

The performance gap is particularly striking for the proton peak at $\delta$~3.71~ppm, where the standard Transformer shows a $\sim$30\% degradation in feature consistency (mean 0.681 vs. 0.971). The collapse is most severe for high-field alkyl protons at $\delta$~0.93~ppm and 0.89~ppm, where the lack of permutation invariance causes the baseline model's similarity to drop to 0.355 and 0.368, respectively, while NMRTrans maintains near-perfect consistency ($\mu \geq 0.984$). These results empirically validate that the ISAB-based architecture effectively decouples chemical semantics from arbitrary peak indexing, satisfying the equivariance property $f(\pi(\mathcal{P})) = \pi(f(\mathcal{P}))$.

\section{Discussion and Future Work}
\label{sec:discussion}

\subsection{Discussion}
\label{subsec:discussion}

First, the performance gap between simulated and experimental spectra is not merely a data distribution issue but stems from \textit{fundamental physical discrepancies}: simulated spectra assume idealized conditions (isolated molecules in vacuum, perfect decoupling), whereas experimental spectra embed solvent effects, dynamic equilibria, and instrument-specific artifacts that cannot be fully captured by current quantum chemistry approximations~\cite{Schreckenbach1995}. 
This explains why even large-scale simulated training fails to transfer robustly—models learn spurious correlations between idealized peak patterns and structures that break down under experimental noise. 
\textbf{\dataset} directly addresses this by providing authentic spectral distributions, but the residual performance gap suggests that \textit{experimental NMR spectra inherently contain ambiguous structural information} for certain molecular classes—a limitation of the spectroscopic technique itself rather than modeling deficiency.

Second, the architectural advantage of Set Transformers is not simply about ``removing positional encodings'' but about respecting the \textit{physical ontology} of NMR data. 
As established in Section~\ref{sec:preliminary}, NMR peaks form a mathematically unordered set because chemical shift values ($\delta$) are intrinsic molecular properties independent of acquisition order. 
Standard Transformers with positional encodings implicitly assume that peak index $i$ carries structural semantics (e.g., ``peak 5 encodes methyl groups''), which contradicts NMR physics: the same molecule measured on different instruments—or even with different pulse sequences on the same instrument—yields peaks in arbitrary order. 
Our ablation study (Table~\ref{tab:ablation_on_architecture}) quantifies this mismatch: the +6.08\% gain from removing positional encodings isolates the penalty of forcing sequential bias onto unordered data. 
Crucially, this penalty amplifies with molecular complexity (Figure~\ref{fig:mol_size}), as spectral congestion in larger molecules increases the likelihood of arbitrary peak reordering during acquisition—precisely where permutation invariance becomes critical.

(1) The 42.81\% Top-1 accuracy ceiling likely reflects genuine information-theoretic constraints of 1D NMR: constitutional isomers with similar carbon frameworks (e.g., regioisomers in polyoxygenated terpenes) often produce near-identical $^{13}$C spectra, making unambiguous distinction impossible without 2D connectivity data. 
(2) Performance degradation beyond 40 heavy atoms stems from exponential growth in isomeric possibilities under fixed elemental composition—a combinatorial challenge rather than modeling deficiency. 
(3) SMILES representation inherently discards stereochemical information that NMR spectra \textit{can} encode (e.g., through $^3J_{\text{HH}}$ coupling constants), creating a representational bottleneck where the model ``knows'' more than it can express.

\subsection{Future Work}
\label{subsec:future_work}

Three extensions address these limitations while respecting practical constraints of experimental chemistry:

\begin{enumerate}[leftmargin=*,noitemsep]
    \item \textbf{Minimal 2D integration for connectivity resolution.} We propose augmenting 1D inputs with sparse 2D constraints: for ambiguous regions identified via uncertainty quantification (e.g., overlapping aliphatic chains), request targeted HSQC/HMBC experiments that provide only critical H--C or long-range H--H connectivities. These cross-peaks can be represented as relational tuples $(\delta_{\text{H}}, \delta_{\text{C}})$ and fused via set attention without architectural overhaul—addressing the isomer ambiguity limitation with minimal experimental overhead.
    
    \item \textbf{Uncertainty-guided experimental design.} Integrating Monte Carlo dropout during inference would quantify prediction uncertainty per spectral region (e.g., high uncertainty in 0.5--2.0 ppm aliphatic zone). This uncertainty map could directly guide chemists toward minimal additional experiments (e.g., ``acquire COSY to resolve ambiguity between C3--C4 vs. C4--C5 coupling''), transforming the model from a passive predictor to an active experimental collaborator.
    
    \item \textbf{Chirality-aware decoding via coupling constraints.} To overcome SMILES' stereochemical blindness, we propose conditioning the decoder on \textbf{experimentally derived} $^3J_{\text{HH}}$ values (extracted from $^1$H spectra) using Karplus relationship priors. During beam search, candidate SMILES violating expected dihedral-angle constraints would be penalized—enabling stereochemical disambiguation without requiring full 3D structure prediction. This leverages existing spectral information rather than demanding new experimental modalities.
\end{enumerate}


\section{GenAI Disclosure}
The authors used large language models  solely for language polishing of the manuscript. All scientific content, analysis, and conclusions were generated and verified by the authors, who take full responsibility for the final text.

\section{Case Study}
\label{app:case_study_visuals}

In this section, we provide additional qualitative comparisons between \mname{} and existing baselines. These visualizations highlight the model's ability to reconstruct precise chemical connectivity from experimental spectral inputs.

\begin{figure*}[htbp]
    \centering
    \includegraphics[width=0.95\linewidth]{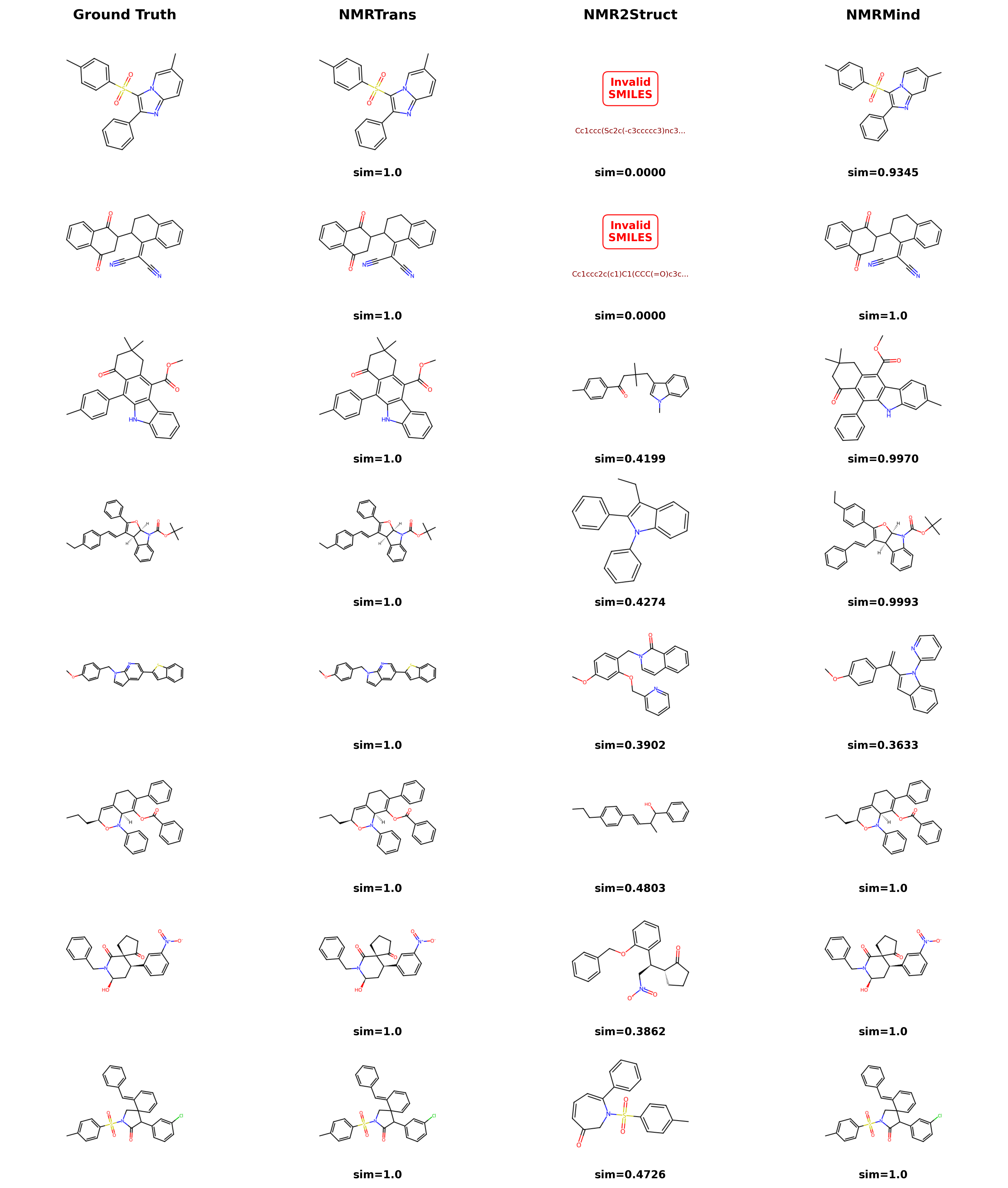}
    \caption{Qualitative comparison of molecular structure elucidation. We compare the ground-truth structures with predictions from \mname{}, NMR2Struct, and NMRMind across diverse chemical classes.}
    \label{fig:app_case_1}
\end{figure*}

\begin{figure*}[htbp]
    \centering
    \includegraphics[width=0.95\linewidth]{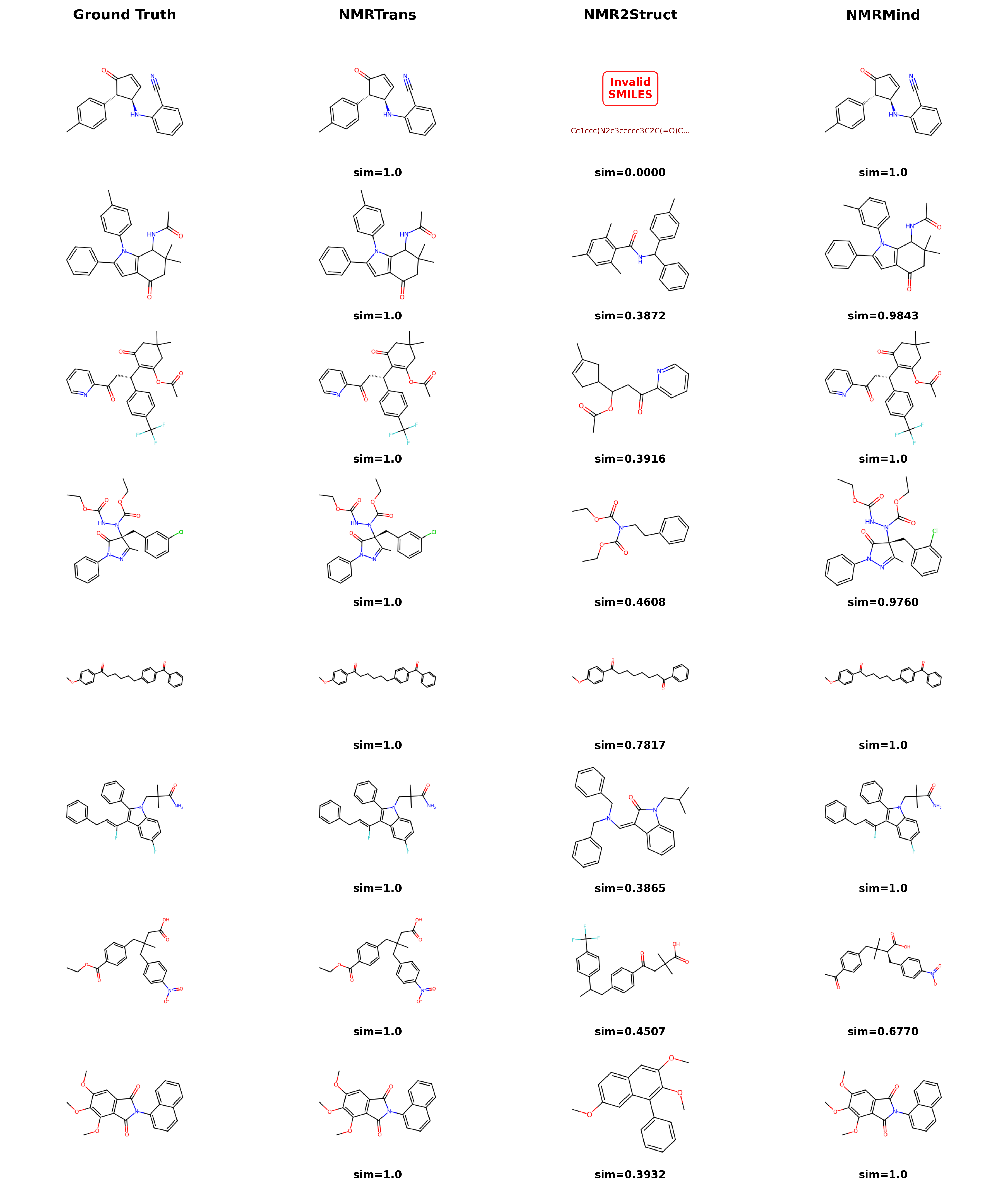}
    \caption{Qualitative comparison of molecular structure elucidation. We compare the ground-truth structures with predictions from \mname{}, NMR2Struct, and NMRMind across diverse chemical classes.}
    \label{fig:app_case_2}
\end{figure*}

\end{document}